\documentclass[]{article}
\usepackage[all]{xy}
\usepackage{amsthm}
\usepackage{amsfonts}
\usepackage{amsmath}
\usepackage{amssymb}
\usepackage{graphicx}
\theoremstyle{definition}
\newtheorem{definition}{Definition}
\newtheorem{proposition}{Proposition}
\newtheorem{lemma}{Lemma}
\newtheorem{theorem}{Theorem}
\newcommand{\R}{\mathbb{R}}
\newcommand{\PP}{\mathbb{P}}
\newcommand{\Q}{\mathbb{Q}}
\newcommand{\T}{\mathbb{T}}
\newcommand{\Z}{\mathbb{Z}}

\newcommand{\N}{\mathbb{N}}

\makeatletter
\let\@fnsymbol\@alph
\makeatother

\title{Robust minimal matching rules for quasicrystals}
\author{Pavel Kalugin\footnote{Laboratoire de Physique des Solides, CNRS, Universit\'e Paris-Sud, Universit\'e Paris-Saclay, F-91405 Orsay, France. E-mail: \texttt{kalugin@lps.u-psud.fr}}
	\and Andr\'e Katz\footnote{Directeur de recherche honoraire, CNRS, France}}

\date{}

\begin{document}

\maketitle

\begin{abstract}
We propose a unified framework for dealing with matching rules of quasiperiodic patterns, relevant for both tiling models and real world quasicrystals. The approach is intended for extraction and validation of a minimal set of matching rules, directly from the phased diffraction data. The construction yields precise values for the spatial density of distinct atomic positions and tolerates the presence of defects in a robust way.
\end{abstract}
\tableofcontents
\section{Introduction}\label{sec:intro}
What is meant by resolving the structure of a quasicrystal? The question is far from rhetorical, since contrarily to the crystals, an aperiodic structure cannot be described by coordinates of a finite set of atoms in a unit cell. This does not mean that an infinite quasiperiodic arrangement of points is impossible to describe finitely. Consider for example the ``cut-and-project'' method\footnote{This construction repeats essentially that of Yves Meyer model sets \cite{meyer1972algebraic}.}, in the situation when the shape of atomic surfaces depends on a finite set of parameters. The number of these parameters, although finite, is potentially unbounded, and this fact raises a perplexing question: how many of them should be used to fit the experimental diffraction data? To solve this conundrum, we have to recall that real quasicrystals are {\em self-assembled} structures and that the assembly is governed by short-range forces. This consideration leads to the notion of {\em matching rules} (that is, local constraints enforcing global aperiodic order). The existence of such rules, for instance, imposes restrictions on possible shapes of atomic surfaces in the cut-and-project models. Namely, in the case of polyhedral surfaces, only faces of rational directions with small denominators are allowed \cite{katz2013number}. The importance of this constraint suggests that matching rules should play a central role in the study of the quasicrystalline structures. Ideally, the answer to the question raised at the beginning should be the following: {\em resolving the structure of a quasicrystal means finding matching rules that enforce a quasiperiodic arrangement of atoms yielding diffraction amplitudes compatible with the experimental data}. The purpose of this paper is to suggest a framework for description of such rules. 
\par
The problem of matching rules is traditionally formulated in terms of tilings, largely for historical reasons. In fact, well before becoming relevant to the solid state physics with the discovery of quasicrystals, the question of matching rules had arisen in computability theory \cite{wang1961proving,berger1966undecidability}. As a result, the language of tilings, subshifts and symbolic dynamics dominated the field ever since. In particular, in the tiling-based structure models of quasicrystals, the arrangement of atoms in the real space is thought of as just a decoration of imaginary rigid tiles, while the aperiodic order is enforced by matching of labels or indentations of some sort on the tiles sides. The resulting aperiodic tiling of the space thus provides an invisible scaffold for the disposition of real atoms. Needless to say that no such hidden structure exists in the real world quasicrystals. Instead,  the role of ``signals'' in the propagation of order from the small to the large scale is fulfilled by the positions and chemical nature of atoms. Thus, even if it is hardly possible to avoid speaking of tilings, the shape of tiles in our model should be implied by the atomic positions. This naturally leads to the choice of simplex-shaped tiles, spanned by the atomic positions at their vertices. We shall also avoid mapping the tiling model on a lattice (although this is a common practice in the mathematical community). Instead, the metric parameters of the tiles will be integrated in the model. 
\par
Matching rules on tiling models can enforce various classes of aperiodic order. In particular, any substitution tiling can be stabilized by matching rules after adding some (not necessarily locally derivable!) decorations \cite{goodman1998matching}. However, the only type of aperiodic order observed so far in solid state matter is the quasiperiodic one. The specificity of quasiperiodic tilings consists in the possibility to lift them in a higher-dimensional space. The tiling then comes out as a projection of a corrugated surface, composed of elements of a higher-dimensional periodic structure \cite{de1981algebraic}. The long-range aperiodic order is characterized by the mean slope of the corrugated surface. This leaves room for integration of an imperfect order in the model, for instance, by allowing this surface to wiggle slightly around its mean direction. This feature is important for modeling the real materials, where defects are always present. However, with few exceptions (e.g., the square-triangle tiling), the lift construction has been so far developed only for the tilings composed by parallelograms or parallelepipeds. We will need to extend this approach to a generic simplex tiling.
\par
In this paper we address the problem of propagation of an aperiodic order with the tools of algebraic topology. A similar geometrical interpretation of the matching constraints has been used in \cite{katz1995matching} in the context of the standard cut-and-project model extended to allow for undulating cuts. In this model, the local matching conditions of the tiling are described by a periodic arrangement of obstacles for the cut, constructed in such a way that the cuts of the same homotopy class produce identical tilings. The obstacles are said to realize the matching rules if every homotopy class of the cut contains a representative of a flat cut. However, in the general case the computation of the homotopy classes can be quite involved. In this paper we show that similar results can be obtained by using the much simpler {\em homology} methods. More importantly, the homology tools turn out to better suit the analysis of the experimental data with their inherent finite precision.
\par
The rest of the paper is organized as follows. In Section \ref{sec:fbs} we introduce the geometric encoding of matching rules in terms of FBS-complexes. In Section \ref{sec:lifting} we construct the lifting of arbitrary simplicial tilings of finite local complexity and define weak and minimal matching rules. Section \ref{sec:main} contains the main results of the paper. In Section \ref{sec:real_qc} we propose a strategy for exploration of matching rules in real quasicrystals and give a sketch of an algorithm following this strategy. The questions remaining open are discussed in Section \ref{sec:conclusion}. 
\section{FBS-complexes}\label{sec:fbs}
Piecewise flat topological spaces emerge frequently in the study of aperiodic order. For instance, if the tiling is defined by a set of matching rules, such objects appear naturally as the so-called prototile spaces \cite{savinien2009spectral}. The latter are obtained as a result of gluing together all prototiles along the matching faces. The prototile space can thus be considered as a way to encode geometrically the combinatorial information about matching rules. Similar constructions are possible when the matching rules are formulated in terms of overlapping clusters.
\par
Alternatively, piecewise flat topological spaces appear as approximations to the tiling space of a given aperiodic pattern, even if the matching rules defining it do not exist or are not known. For instance, for patterns of finite local complexity (FLC) \cite{baake2013aperiodic}, one can define the following equivalence relation between points of the underlying space: the points $x$ and $y$ are equivalent if the parts of the patterns contained within {\em open} discs of some fixed radius $r$ centered at $x$ and $y$ coincide modulo translation by $x-y$. The set of equivalence classes with the quotient topology is a compact Hausdorff space homeomorphic to a finite CW-complex \cite{walton2014}. A similar object in the discrete setting is known as G\"ahler's collared tiles construction (see e.g., \cite[p. 84]{sadun}).
\par
The spaces obtained by either of the above constructions are invariant (up to a homeomorphism) with respect to deformations or redecorations of the tiling within its mutual local derivability (MLD) \cite{baake2013aperiodic} class. This is an indication that the piecewise flat spaces could be considered independently on the underlying aperiodic patterns. The characterization of such spaces as just finite CW-complexes neglects the metric structure inherited from the underlying space. This issue is specifically addressed by the construction of branched oriented flat manifolds, proposed in \cite{bellissard2006spaces}. Yet these objects, described by local models, are unnecessarily intricate for our purposes. Instead, we introduce a model based on simplicial complexes \cite{spanier2012algebraic}, since among various cellular structures used in the combinatorial topology, they are the easiest to equip with the metric information. We have, however, to relax the requirement for each simplex to be uniquely defined by the set of its vertices:
\begin{definition} \label{FBS}
	A $d\mbox{-dimensional}$ flat-branched semi-simplicial complex (FBS-complex) $B$ is a finite connected topological semi-simplicial\footnote{We use here the original terminology, as it was first introduced in \cite{eilenberg1950semi}. This construction is also known under the name of $\Delta\mbox{-set}$ \cite{rourke1971delta} or $\Delta\mbox{-complex}$ \cite[Chap.~2]{hatcher2002algebraic}.} complex of dimension $d$ equipped with a homomorphism $\rho: C_1(B) \to E$ of the group of 1-chains of $B$ to a $d\mbox{-dimensional}$ real Euclidean vector space $E$, satisfying the following conditions:
	\begin{itemize}
		\item The homomorphism $\rho$ vanishes on boundaries: 
		\begin{equation}
		\label{rho_d}
	    \rho \circ \partial=0.
		\end{equation}
		\item For any $k\mbox{-simplex}$ $s \in B$ and the set $\{e_1, \dots, e_k\}$ of edges of $s$ originating in the same vertex, the vectors $\rho(e_1), \dots, \rho(e_k)$ are linearly independent.  
	\end{itemize} 
\end{definition}
While in general semi-simplicial complexes a given simplex can be glued to itself by faces of any dimension, the second condition of the Definition \ref{FBS} allows only for gluing by vertices:
\begin{proposition}\label{notglued}
	For any $k\mbox{-simplex}$ of an FBS-complex with $k>0$, all edges are distinct (that is, the corresponding iterated face maps have $k(k+1)/2$ distinct images).
\end{proposition}  
\begin{proof}
	Any two edges of a $k\mbox{-simplex}$ are contained in one of its 2-dimensional or 3-dimensional faces. If two edges are glued together, then in the first case, there exists a $2\mbox{-simplex}$ in $B$ with edges $e_1$ and $e_2$ originating in the same vertex and having $\rho(e_1)=\rho(e_2)$. In the second case $B$ contains a $3\mbox{-simplex}$ with edges $e_1$, $e_2$ and $e_3$ originating in the same vertex, such that $\rho(e_1)=\rho(e_2)-\rho(e_3)$. In either case, the condition of linear independence is not satisfied.
\end{proof}
\par
We shall make a distinction between purely combinatorial objects and their geometric realization. For the latter we shall use the traditional ``vertical bars'' notation, for example, $|B|$ will denote the geometric realization of the abstract complex $B$. For a $k\mbox{-simplex}$ $s\in B$, let $\{e_1, \dots, e_k\}$ be the edges of $s$ originating in the same vertex. For an arbitrarily chosen point $x \in E$ one can construct an affine $k\mbox{-simplex}$ $\sigma \subset E$ with vertices 
$$\left\{x, x+\rho(e_1), \dots, x+\rho(e_k)\right\}.$$ 
Identification of barycentric coordinates in both simplices defines a homeomorphism of $|s|$ onto the interior of $\sigma$: 
\begin{equation}
\label{alpha_s}
\alpha_s: |s| \to \sigma^\circ.
\end{equation}
Since the choice of $x$ is arbitrary, $\alpha_s$ is defined up to a translation in $E$; the property (\ref{rho_d}) guarantees that this definition does not depend on the ordering of vertices in $s$. If $s$ is a $d\mbox{-dimensional}$ simplex, we shall refer to $\overline{\alpha_s(|s|)}$ as a {\em prototile}. Note that the Euclidean structure of $E$ can be pulled back by $\alpha_s$ to $|s|$; this justifies the use of the term ``flat'' applied to $B$.
\par
Recall now that the FBS-complexes were introduced as a way to encode the structure of aperiodic tilings and the corresponding matching rules. Within this framework, tilings are described by a certain class of maps from $E$ to $|B|$, respecting the local Euclidean structure of both spaces:
\begin{definition}
	An {\em isometric winding} of an FBS-complex $B$ is a continuous map $f: E  \to |B|$ such that
	\begin{itemize}
		\item The full preimage of any open $d\mbox{-dimensional}$ simplex $|s| \subset |B|$ is a disjoint union of interiors of affine simplices $\sigma_{s,i}$:
		\begin{equation}
		\label{preimage}
		f^{-1}(|s|)=\bigsqcup_i \sigma_{s,i}^\circ, \qquad \sigma_{s,i}^\circ \cap \sigma_{s,j}^\circ = \emptyset \text{ if } i \neq j 
		\end{equation}
		which are translated copies of the corresponding prototile:
		\begin{equation}
		\label{translated_prototile}
		\alpha_s(f(x))=x+\tau_{s,i}\qquad \forall x \in \sigma_{s,i}^\circ
		\end{equation}
		for some $\tau_{s,i} \in E$.
		\item The simplices $\sigma_{s,i}$ cover the entire space $E$:
		\begin{equation}
		\label{domain}
		E = \bigcup_{s,i} \sigma_{s,i}
		\end{equation}
	\end{itemize}
\end{definition} 
\begin{proposition}\label{from_f_to_T}
	If $f: E \to |B|$ is an isometric winding, then the covering of $E$ by affine simplices of the form (\ref{preimage}) is a tiling $\mathcal{T}$ of $E$ by translated copies of prototiles. If two tiles $\sigma_i$ and $\sigma_j$ in $\mathcal{T}$ share a common face, then the same is true for the simplices $f(\sigma_i^\circ)$ and $f(\sigma_j^\circ)$ in $|B|$. The tiling $\mathcal{T}$ has finite local complexity (FLC) with respect to translations and the vertices of $\mathcal{T}$ form a Delone\footnote{Delone is a common transliteration for the name of B.N.Delaunay (as in ``Delaunay triangulation''.)} set in $E$ (see \cite{baake2013aperiodic} for the definitions of FLC and Delone properties). 
\end{proposition}
\begin{proof}
	Since the entire space $E$ is covered by tiles $\sigma_{s,i}$ and their interiors do not intersect because of (\ref{preimage}), $\mathcal{T}$ is a tiling of $E$. The statement about the common face follows from the inclusion
	$$
	f(\sigma_i \cap \sigma_j) \subset f(\sigma_i) \cap f(\sigma_j).
	$$
	To prove the FLC property, consider a ball $\mathcal{B}_r$ of radius $r$ centered at a vertex of $\mathcal{T}$. Since there is a finite number of prototiles and all of them have non-empty interior, the number of tiles contained within $\mathcal{B}_r$ is bounded by some constant, depending only on $r$. Then the coordinates of the vertices of $\mathcal{T}$ within the ball $\mathcal{B}_r$ with respect to its center are linear combinations of a finite set of vectors $\{\rho(e_i), e_i \in \mbox{edges}(B)\}$ with bounded integer coefficients. Therefore, the set of vertices of $\mathcal{T}$ has finite local complexity with respect to translations. The same is true for the tiling itself since the number of prototiles is finite. 
    \par
    To prove the Delone property, let us choose $r_1>0$ to be smaller than half of the smallest distance from a vertex of a prototile to the opposite face. Then any two balls of radius $r_1$ centered at different vertices of $\mathcal{T}$ have zero intersection. Similarly, if $r_2$ is larger than the longest edge of every tile, any point in $E$ lies at the distance smaller than $r_2$ from a vertex of $\mathcal{T}$. Therefore, the set of vertices of $\mathcal{T}$ is uniformly discrete and relatively dense.
\end{proof}
\begin{proposition}\label{from_T_to_f}
	Let $B$ be a $d\mbox{-dimensional}$ FBS-complex and $\mathcal{T}$ be a tiling of the $d\mbox{-dimensional}$ Euclidean space $E$ by translated copies of prototiles $\{\overline{\alpha_s(|s|)}, s \in B\}$, such that the tiles $\sigma_i$ and $\sigma_j$ in $\mathcal{T}$ either have no intersection of dimension $d-1$, or share a common face, in which case so do the corresponding simplices in $B$. Then there exists an isometric winding $f: E \to |B|$ such that the tiling constructed in Proposition \ref{from_f_to_T} coincides with $\mathcal{T}$.
\end{proposition}
\begin{proof}
	Any point $x \in E$ is contained in at least one tile $\sigma$ of $\mathcal{T}$. Let $\overline{\alpha_s(|s|)}$ be the corresponding prototile. We set for $f(x)$ the point of $|s|$ having the same barycentric coordinates as $x$ within $\sigma$. If $x$ belongs to two different tiles, these tiles share a common face and the same is true for the corresponding simplices of $B$. Thus the definition of $f(x)$ does not depend on the choice of the simplex containing $x$. By induction the same holds if $x$ belongs to several tiles, hence the map $f: E \to |B|$ is well-defined and continuous. By its construction $f$ satisfies the properties (\ref{preimage}) and (\ref{translated_prototile}). Therefore $f$ is an isometric winding, and the tiles in the corresponding tiling coincide with those of $\mathcal{T}$.  
\end{proof}
The Propositions \ref{from_f_to_T} and \ref{from_T_to_f} demonstrate that the matching rules for any FLC tiling by simplices can be encoded by an appropriate FBS-complex. Since any polyhedron can be triangulated by simplices, this means that any problem of matching rules for FLC tiling by polyhedral tiles (and then, according to \cite{kenyon1992rigidity}, by any topological disks) can be formulated in terms of FBS-complexes and isometric windings. Therefore, by the classic Berger's result \cite{berger1966undecidability}, the problem of existence of an isometric winding for a given FBS-complex $B$ is undecidable. More precisely, there is no regular way to prove that there exists an isometric winding for an arbitrary $B$; although in some cases such proof is possible --- e.g., when the corresponding tiling is periodic or can be obtained by substitutions. The opposite statement, though, can be verified algorithmically in a finite (but unpredictably long) time. For the practical purposes one can simplify the problem before attempting to prove the non-existence of an isometric winding:
\begin{proposition}\label{reduction}
	Let $s \in B$ be a simplex of dimension $k<d$ and $\sigma=\overline{\alpha_s(|s|)}$. Consider the finite set $\{\sigma_i\}$ of translated copies of prototiles having $\sigma$ as a $k\mbox{-face}$ (if $k=0$, this subset may contain several copies of the same prototile). If no subset of $\{\sigma_i\}$ represents a tiling of a polyhedral disc in $E$, containing $\sigma^\circ$ in its interior, then $|s|$ may not belong to the image of any isometric winding. If this is the case, $B$ admits an isometric winding if and only if an isometric winding exists for a reduced FBS-complex $B'$, obtained from $B$ by removing all simplices having $s$ as a $k\mbox{-face}$.
\end{proposition}
\begin{proof}
	The first part of the statement follows from the observation that for an isometric winding having $|s|$ in its image, the corresponding tiling $\mathcal{T}$ contains a translated copy $\hat\sigma$ of $\sigma$ and the set of tiles in $\mathcal{T}$ having $\hat\sigma$ as a $k\mbox{-dimensional}$ face form a tiling of a disc in $E$, such that $\hat\sigma^\circ$ is contained in its interior. For the second part of the statement, note that the natural inclusion $|B'|\subset |B|$ makes an isometric winding of $B'$ also an isometric winding of $B$. On the other hand, an isometric winding of $B$, not containing $|s|$ in its image does not contain any of simplices of $|B|$ having $|s|$ as a $k\mbox{-face}$, and is therefore also an isometric winding of $B'$.
\end{proof}
\par
Each simplicial tile in $\mathcal{T}$ is entirely defined by its vertices. Therefore, the covering of $E$ by tiles and their faces of all dimensions represents a geometric realization of an infinite simplicial complex. We shall use the same symbol $\mathcal{T}$ to denote the corresponding abstract complex, and interpret an isometric winding $f$ as a semi-simplicial map $\mathcal{T} \to B$. In particular, we can speak of an edge-path $a\dots b$ between two vertices $a$ and $b$ of $\mathcal{T}$.
\begin{proposition}
	Let $f: E \to |B|$ be an isometric winding corresponding to the tiling $\mathcal{T}$. Then if the points $a, b \in E$ are vertices of $\mathcal{T}$, for any edge-path $a\dots b$, the following holds:
	$$
	\rho(f(a \dots b))=b-a.
	$$
\end{proposition}
\begin{proof}
	By linearity of $\rho$ it suffices to prove the statement for the edges of $\mathcal{T}$. Let $[a, b]$ be an edge of a tile $\sigma$ in $\mathcal{T}$ and $\alpha_s(|s|)$ be the corresponding prototile. Then $f([a, b])$ is an edge of $s$ and the corresponding edge in $\alpha_s(|s|)$ is given by 
	$$
	[x, x+\rho(f([a, b]))]
	$$ 
	for some $x \in E$. By (\ref{translated_prototile}) this line segment is a translation of $[a, b]$, therefore
	$$
	\rho(f([a, b]))=b-a.
	$$
\end{proof}
Let us fix once and for all an orientation of $E$. The corresponding orientation of prototiles is pulled back to the simplices of $B$ by the maps $\alpha_s$ and to the simplices of $\mathcal{T}$ by translations. The choice of the positively oriented $d\mbox{-simplices}$ as the basis defines the positive cones $C_d^+(B)\subset C_d(B)$ and $C_d(\mathcal{T})^+ \subset C_d(\mathcal{T})$ and the corresponding partial order on the groups of $d\mbox{-chains}$.
\begin{proposition}\label{orientation}
	Let $f: E \to |B|$ be an isometric winding corresponding to the tiling $\mathcal{T}$. Then $f_*:C_d(\mathcal{T}) \to C_d(B) $ preserves the partial order:
	$$
	f_*(C_d(\mathcal{T})^+) \subset C_d^+(B).
	$$
\end{proposition}
\begin{proof}
	The statement follows from the observation that isometric windings preserve orientation of simplices by the property (\ref{translated_prototile}).
\end{proof}
\par
Restricting $\rho$ to 1-cycles defines the homomorphism $\rho_*: H_1(B) \to E$ (the property (\ref{rho_d}) makes the definition independent on the choice of the representative cycle). The image of $\rho_*$ is a finitely generated free abelian subgroup of $E$, which we will denote by $L$. Let us show now that the existence of an isometric winding implies that $L$ spans $E$:
\begin{proposition}\label{spanning}
If an FBS-complex admits an isometric winding, then $L\otimes_\Z \R= E$
\end{proposition}
\begin{proof}
	Let us proceed by {\em reductio ad absurdum}. Assume that $B$ admits an isometric winding $f$, but $L\otimes_\Z \R \subsetneq E$. Then for any vertex $v$ of $B$, the set $f^{-1}(v)$ is either empty or
	$$
	f^{-1}(v) \in L\otimes_\Z \R + x
	$$ 
	for some $x \in f^{-1}(v)$. Indeed, for any two points $a, b \in f^{-1}(v)$, one can construct an edge-path $a\dots b$ on the simplicial tiling $\mathcal{T}$ of $E$ corresponding to $f$. Its image $f(a\dots b)$ is a cycle in $B$, therefore $b-a \in L \subset L\otimes_\Z \R$. Then, the full set of vertices of $\mathcal{T}$ is contained in a finite union of proper subspaces of $E$ and thus cannot be relatively dense, which contradicts its Delone property. 
\end{proof}
Let us now introduce a class of maps between the geometric realizations of FBS-complexes, preserving their Euclidean structure. We shall start by considering the special case of simplex-to-simplex maps.
\begin{definition}\label{simplicial-FBS-map}
    A continuous map $\epsilon: |B_2| \to |B_1|$ is called a {\em simplicial FBS-map} if $\epsilon$ takes 
	each $d\mbox{-simplex}$ $|s_2|$ of $|B_2|$ to a $d\mbox{-simplex}$ $|s_1|$ of $B_1$ and 
	$$
	\alpha_{s_2}(x)=\alpha_{s_1}(\epsilon(x)) + \tau_{s_1, s_2} \qquad \forall x \in |s_2|
	$$
	for some $\tau_{s_1, s_2}\in E$.
\end{definition}
\begin{proposition}\label{simplicial-FBS-iw}
	If $\epsilon: |B_2| \to |B_1|$ is a simplicial FBS-map, then for any isometric winding $f_2$ of $B_2$, the map $f_1=\epsilon \circ f_2$ is an isometric winding of $B_1$.
\end{proposition}
\begin{proof}
	Consider a $d\mbox{-simplex}$ $s_1 \in B_1$ and its preimage in $B_2$:
	$$
	S_{s_1}=\left\{s_2 \in B_2 \mid \epsilon(|s_2|)=|s_1|\right\}.
	$$
	Then, since $f_2$ is an isometric winding,
	$$
	{f_1}^{-1}(|s_1|)=\bigcup_{s_2 \in S_{s_1}} {f_2}^{-1}(|s_2|)=
	\bigsqcup_{s_2 \in S_{s_1}} \bigsqcup_{i} \sigma_{s_2, i}^\circ
	$$
	hence $f_1$ satisfies the condition (\ref{preimage}). The condition (\ref{translated_prototile}) holds as well since for any tile $\sigma_{s_2, i}^\circ$ we have
	$$ 
	\alpha_{s_1} (f_1(x)) =
	\alpha_{s_1}(\epsilon(f_2(x))) =
	\alpha_{s_2}(f_2(x)) + \tau_{s_1, s_2}\qquad \forall x \in \sigma_{s_2, i}^\circ.
	$$
	Finally, the property (\ref{domain}) for $f_1$ follows immediately from that for $f_2$. Thus, $f_1$ is an isometric winding of $B_1$.
\end{proof}
To extend Definition \ref{simplicial-FBS-map} to the situation when maps are not simplex-to-simplex, we need subdivisions of FBS-complexes:
\begin{definition}\label{subdivision}
	If $B$ and $B'$ are FBS-complexes, a homeomorphism $\varsigma: |B'| \to |B|$ of their geometric realizations is called a subdivision of $B$ if it is a subdivision in the sense of CW-complexes and if it respects the Euclidean structure of the cells, that is whenever $\varsigma(|s'|)\subset |s|$ for $d\mbox{-simplices}$ $s \in B$ and $s' \in B'$, one has
	\begin{equation}
	\label{subdiv}
	\alpha_{s'}(x)= \alpha_s(\varsigma(x)) + \tau_{s,s'}\qquad \forall x \in |s'|
	\end{equation}
	for some $\tau_{s,s'} \in E$.
\end{definition}
\begin{proposition}\label{subdiv-iw}
	If $\varsigma: |B'| \to |B|$ is a subdivision of an FBS-complex $B$, then the map
	$f': E \to |B'|$ is an isometric winding of $B'$ if and only if $f=\varsigma \circ f'$ is an isometric winding of $B$.
\end{proposition}
\begin{proof}
	For a $d\mbox{-simplex}$ $s$ of $B$, let $S_s$ stand for the subset of $d\mbox{-simplices}$ of $B'$ contained in $\varsigma^{-1}(s)$. 
	Since the prototile map $\alpha_{s'}$ is defined up to translation in $E$, one can always choose $\tau_{s, s'}=0$ in (\ref{subdiv}):
	\begin{equation}
	\label{subdiv_bis}
	\alpha_{s'}= \alpha_s\circ\varsigma\big|_{|s'|}
	\end{equation}
	Then  $\alpha_s(|s|)$ is partitioned into a finite set of open subtiles $\{\alpha_{s_j'}(|s_j'|) \mid s_j' \in S_s\}$:
	\begin{equation}
	\label{tessellation}
	\overline{\alpha_s(|s|)}=\overline{\bigsqcup_{s_j' \in S_s} \alpha_{s_j'}(|s_j'|)}
	\end{equation}
	Since $\alpha_s$ and $\varsigma$ are homeomorphisms, two subtiles in (\ref{tessellation}) share a common face of dimension $d-1$ if and only if so do the corresponding simplices from $S_s$.
	\par
	Let $f': E \to |B'|$ be an isometric winding of $B'$. Denote the corresponding tiling of $E$ by $\mathcal{T'}$. Then according to (\ref{preimage}) for any simplex $s' \subset S_s$
	$$
	f'^{-1}(|s'|) = \bigsqcup_i \sigma_{s', i}^\circ,
	$$
	where $\sigma_{s', i}$ are tiles of $\mathcal{T}'$. Consider a simplex $s'_0\in S_s$ and one of the corresponding open tiles $\sigma_{s'_0, i}^\circ$. Then (\ref{translated_prototile}) defines the translation $\tau_{s'_0,i}$ such that
	$$
	\sigma_{s'_0, i}^\circ = \alpha_{s'_0}(|s'_0|)-\tau_{s'_0,i}
	$$ 
	Let us show now that the same translation $\tau_{s'_0,i}$ takes other subtiles of the partitioning (\ref{tessellation}) exactly to the interiors of tiles of $\mathcal{T}'$. We shall assign to these tiles the same tile index $i$ and denote them by $\sigma_{s'_j, i}$. One can proceed by induction over $S_s$. Consider $d\mbox{-simplices}$ $s'_j, s'_k \in S_s$ sharing a common face of dimension $d-1$ and assume that $\alpha_{s'_j}(|s'_j|) -\tau_{s'_0,i}=\sigma_{s'_j, i}^\circ$. Then, since $s_k'$ is the only $d\mbox{-simplex}$ of $B'$ sharing this face with $s_j'$, we also have $\alpha_{s'_k}(|s'_k|) -\tau_{s'_0,i}=\sigma_{s'_k, i}^\circ$. Since $\alpha_s(|s|)$ is a convex subset of $E$, any simplex from $S_s$ can be reached from $s'_0$ by following a sequence of $d\mbox{-simplices}$ in which any two consecutive elements share a common face of dimension $d-1$. Therefore
	$$
	\alpha_{s'_j}(|s'_j|) -\tau_{s'_0,i}=\sigma_{s'_j, i}^\circ \quad \text{ for all }s'_j \in S_s.
	$$
	We can now construct a supertile
	\begin{equation}
	\label{supertile}
	\sigma_{s, i}=\bigcup_{s'_j \in S_s} \sigma_{s'_j, i}.
	\end{equation}
	Since $f$ is continuous, 
	\begin{equation}
	\label{supertile_bis}
	f(\sigma_{s,i}^\circ)=|s|.
	\end{equation}
	The identity (\ref{subdiv_bis}) yields the property (\ref{translated_prototile}) with $\tau_{s,i}=\tau_{s'_0, i}$. Since $\mathcal{T'}$ covers the entire space $E$, the property (\ref{domain}) for $\sigma_{s,i}$ follows from (\ref{supertile}). Finally, as follows from (\ref{supertile}) and (\ref{supertile_bis}), different open supertiles $\sigma_{s,i}^\circ$ do not intesect each other and
	$$
	f^{-1}(|s|) \supset \bigsqcup_i \sigma_{s,i}^\circ .
	$$
	On the other hand, $f^{-1}(|s|)$ is open in $E$ and cannot meet the interior of the supertiles corresponding to the simplices of $B$ other than $s$. Therefore, the condition (\ref{preimage}) holds as well and $f$ is an isometric winding of $B$.
	\par
	Let now $f: E \to |B|$ be an isometric winding and $f'=\varsigma^{-1} \circ f$. Consider a $d\mbox{-simplex}$ $s' \in B'$ and let $s \in B$ be the (unique) $d\mbox{-simplex}$ $s \in B$ such that $\varsigma(|s'|) \subset |s|$. As follows from (\ref{preimage}) and (\ref{translated_prototile}),
	$$
	f^{-1}(|s|)=\bigsqcup_i \left(\alpha_s(|s|)-\tau_{s,i}\right)\quad\text{ where } \tau_{s,i} \in E.
	$$ 
	Define the open tiles $\sigma_{s',i}^\circ$ as
	$$
	\sigma_{s',i}^\circ=\alpha_{s'}(|s'|)-\tau_{s,i}
	$$
	Then (\ref{subdiv_bis}) yields
	\begin{equation}
	\label{f_prim_disjoint}
	f'^{-1}(|s'|)=f^{-1}(\varsigma(|s'|))=\bigsqcup_i \sigma_{s', i}^\circ,
	\end{equation}
	as well as
	\begin{equation}
	\label{translated_subtile}
	\alpha_{s'}(f'(x))= \alpha_{s}(f(x))=
	x + \tau_{s,i} \qquad \forall x \in \sigma_{s',i}^\circ.
	\end{equation}
	Finally
    \begin{equation}
    \label{f_prim_covers}
    \bigcup_{s', i}\sigma_{s',i}=
    \bigcup_{s, i}\left(\overline{\bigsqcup_{s'_j\in S_s}\alpha_{s'}(|s'|)}-\tau_{s,i}\right)=
    \bigcup_{s,i}\sigma_{s,i}=E.
    \end{equation}
    It follows from (\ref{f_prim_disjoint}), (\ref{translated_subtile}) and (\ref{f_prim_covers}) that $f'$ is an isometric winding of $B'$.
\end{proof}
\begin{definition}\label{FBS-map}
	A continuous map $\epsilon: |B_2| \to |B_1|$ between the geometric realizations of FBS-complexes $B_1$ and $B_2$ is called an {\em FBS-map} if there exist subdivisions $\varsigma_1: |B'_1| \to |B_1|$ and $\varsigma_2: |B'_2| \to |B_2|$ such that $\epsilon'=\varsigma_1^{-1}\circ \epsilon \circ \varsigma_2$ is a simplicial FBS-map.
\end{definition}
\begin{proposition}\label{FBS-iw}
	If $\epsilon: |B_2| \to |B_1|$ is an FBS-map and $f_2: E \to |B_2|$ is an isometric winding of $B_2$, then $f_1=\epsilon\circ f_2$ is an isometric winding of $B_1$.
\end{proposition}
\begin{proof}
	The result follows immediately from Propositions \ref{simplicial-FBS-iw} and \ref{subdiv-iw}.
\end{proof}
It is straightforward to prove that compositions of FBS-maps are also FBS-maps, which makes FBS-complexes with FBS-maps into a category.
\section{Lifting of simplicial tilings}\label{sec:lifting}
The idea to describe aperiodic tilings as projections, first suggested by de Bruijn \cite{de1981algebraic}, is commonly used in the study of both the matching rules \cite{bedaride2015periodicities,levitov1988local} and the problems of random tilings \cite{shaw1991long}. Henley \cite{henley1999random} proposed a generalization of this construction to the case of arbitrarily shaped polyhedral tiles with the lifting dimension equal to the rank of the free abelian group generated by the ``linkage vectors'' (which are essentially the edges of the prototiles). However, the lifting produced by this approach has several drawbacks. Consider, for instance, the formal modification of a two-dimensional tiling, consisting in the insertion of a new vertex in a tile edge at a generic position. 
This operation does not change the structure of the tiling, but creates a new linkage vector and thus requires adding an extra dimension to the lifting space. On the other hand, relevant lifting dimensions may be missed in the case of an accidental integral linear dependence in the set of linkage vectors. This happens, for example, in the case of the tiling of plane by 60 degrees rhombi. In this section we revisit the issue of lifting of aperiodic tilings and suggest a different approach to the problem.
\par
In lifting schemes, an aperiodic tiling appears as a projection of a corrugated continuous hypersurface, composed of facets belonging to a larger {\em periodic} pattern. Let us denote the corresponding lattice of periods by $\mathcal{L}$. Since the facets are locally connected in the same way as are the prototiles in the prototile space, this pattern is a covering space of the latter, with monodromy group $\mathcal{L}$. In the case of simplicial tilings the base of the covering is an FBS-complex. Since $\mathcal{L}$ is free abelian, the corresponding coverings are classified by homomorphisms of the first integral homology of the base onto $\mathcal{L}$ \cite{dwyer1987homology}:
\begin{proposition}\label{M_covering}
	Let $B$ be an FBS-complex and let $h: \pi_1(B, b) \to H_1(B)$ stand for the Hurewicz homomorphism \cite{fomenko2016homotopical} (where $b \in B$ is an arbitrary vertex chosen as the base point). Then for any surjective homomorphism $\lambda: H_1(B,\Z) \to \mathcal{L}$ there exists a normal semi-simplicial covering $p:\tilde{B} \to B$ with monodromy group  $\mathcal{L}$.
\end{proposition}
\begin{proof}
	Since the space $|B|$ is path-connected, the Hurewicz map is surjective and $\lambda(h(\pi_1(B, b)))=\mathcal{L}$. Therefore, $\mathcal{L}$ is a quotient of $\pi_1(B, b)$ by its normal subgroup $\ker(\lambda \circ h)$, and by the fundamental theorem of covering spaces there exists a normal covering of $|B|$ having $\mathcal{L}$ as monodromy group. Endowing this topological space with the induced structure of a semi-simplicial complex yields a semi-simplicial covering $p:\tilde{B} \to B$.
\end{proof}
\begin{proposition}\label{right-exact}
	The sequence of abelian groups
	$$
	\xymatrix{
		H_1(\tilde{B}) \ar[r]^{p_*}&  H_1(B) \ar[r]^(0.6){\lambda} &  \mathcal{L} \ar[r] & 0.
	}
	$$
	is exact.
\end{proposition}
\begin{proof}
	This sequence is the abelianization of the following short exact sequence of groups defined by the construction of the covering $p: \tilde{B} \to B$
	$$
	\xymatrix{
	\pi_1(\tilde{B}, \tilde{b}) \ar[r]^{p_*} & \pi_1(B, b) \ar[r]^(0.6){\lambda\circ h} & \mathcal{L}
    },
	$$
	where $\tilde{b} \in \tilde{B}$ and $b \in B$ are arbitrarily vertices chosen as the base points. The results then follows from the fact that the abelianization functor is right exact.
\end{proof}
\begin{proposition}\label{lift_f}
	If $f: E \to |B|$ is an isometric winding of an FBS-complex $B$, then there exists a continuous map $\tilde f: E \to |\tilde{B}|$ such that $p \circ \tilde f=f$. 
\end{proposition}
\begin{proof}
	Since $E$ is contractible and any constant map $E \to |B|$ can be lifted to $|\tilde{B}|$, such map exists by the homotopy lifting property of covering spaces. 
\end{proof}
Note that the lift $\tilde f$ constructed in Proposition \ref{lift_f} is not uniquely defined. More precisely, all such lifts are obtained by the composition of $\tilde f$ with the action of the deck transformation group of the covering $p: \tilde{B} \to B$. 
\par
To finish the construction of the lifting, we need a left inverse of $\tilde f$. Such map exists only if closed edge-paths in $\tilde B$ correspond to zero translations of $E$, that is if $\ker(\lambda) \subseteq \ker(\rho_*)$. Therefore, there should exist a surjective homomorphism $\pi_L: \mathcal{L} \to L$ completing the following commutative diagram:
\begin{equation}
\label{pi_L}
\xymatrix{
	H_1(B, \Z) \ar[rd]_{\lambda} \ar^{\rho_*}[rr] & & L\\
	& \mathcal{L} \ar[ru]_{\pi_L}
}
\end{equation}  
Let us show now that conversely, if the homomorphism $\pi_L$ in (\ref{pi_L}) exists, then isometric windings of $B$ and associated tilings of $E$ can be lifted into the real vector space $V$ defined as 
$$
V=\mathcal{L} \otimes_\Z \R
$$
with the corresponding projection $\pi_E: V \to E$:
$$
\pi_E= \pi_L \otimes_\Z \R
$$
(Note that $\pi_E$ is surjective by Proposition \ref{spanning}).
\begin{proposition} \label{prop_mu}
	There exists a continuous map $\mu: |\tilde{B}| \to V$ satisfying the following properties:
	\begin{itemize}
		\item For any edge-path $v_1\dots v_2$ on $\tilde{B}$
		\begin{equation}
		\label{mu_prop_1}
		\pi_E(\mu(|v_2|)-\mu(|v_1|))=\rho(p(v_1\dots v_2)).
		\end{equation}
		\item For any two simplices $t_1, t_2 \in \tilde{B}$ such that $p(t_1)=p(t_2)$ there exists a lattice translation $l \in \mathcal{L}$ such that
		\begin{equation}
		\label{mu_prop_2}
		\mu(|t_1|)=\mu(|t_2|)+l.
		\end{equation}
	\end{itemize}
\end{proposition}
\begin{proof}
	Let us first construct a map of $\Z\mbox{-modules}$ $\tilde{\rho}: C_1(B) \to V$ such that $\pi_E \circ \tilde\rho = \rho$ and the restriction of $\tilde\rho$ on the group of 1-cycles $Z_1(B) \subset C_1(B)$ coincides with $\lambda$ on the corresponding homology classes. The second condition unambiguously defines the continuation of $\tilde{\rho}$ to the free $\Z\mbox{-module}$ $N=C_1(B) \bigcap \left( Z_1(B) \otimes \Q \right)$. The quotient $C_1(B)/N$ is a free $\Z\mbox{-module}$ and $C_1(B)$ splits (not naturally) as $C_1(B)=N \oplus A$ where $A \simeq C_1(B)/N$. Since $\pi_E$ is surjective and $A$ is free, one can always define $\tilde\rho$ on $A$ in such a way that $\pi_E \circ \tilde\rho\big|_A =\rho \big|_A$. This extends $\tilde\rho$ to the entire $\Z\mbox{-module}$ $C_1(B)$. 
	\par
	We shall construct the map $\mu: |\tilde{B}| \to V$ by defining it on vertices of $|\tilde{B}|$ and then continue it to the interiors, using barycentric coordinates. Let us start by choosing a vertex $v_0 \in \tilde{B}$ and a point $y_0 \in V$ and setting 
	$$
	\mu: |v_0| \mapsto y_0.
	$$
	Then for any other vertex $v \in \tilde{B}$ we set
    \begin{equation}
    \label{def_mu}
	\mu: |v| \mapsto y_0+\tilde\rho(p(v_0\dots v)),
    \end{equation}
	where $v_0\dots v$ is an edge-path on $\tilde{B}$ connecting $v_0$ and $v$. To check that this expression is well defined, consider two edge-paths $c_1$ and $c_2$ connecting $v_0$ with $v$. In this case, the chain $c_1-c_2$ is a cycle and $\lambda(p(c_1-c_2))=0$ by Proposition \ref{right-exact}. On the other hand, since $\tilde\rho$ and $\lambda$ coincide on cycles, $\tilde\rho(p(c_1-c_2))=0$ and the expression (\ref{def_mu}) does not depend on the choice of the edge-path $v_0\dots v$. Finally, as $\pi_E \circ \tilde\rho = \rho$,
	$$
	\pi_E(\mu(|v|))=\pi_E(y_0)+\rho(p(v_0\dots v)),
	$$
	which yields (\ref{mu_prop_1}).
	\par
	To check the property (\ref{mu_prop_2}) it suffices to verify it on all edges of $\tilde{B}$. Let $[a_1b_1]$ and $[a_2b_2]$ be the edges of $\tilde{B}$ such that $p([a_1b_1])=p([a_2b_2])$. Consider an edge-path $a_1\dots a_2$. Since $p(b_1a_1\dots a_2b_2) = p(a_1\dots a_2) \in Z_1(B)$, 
	$$
	\mu(|a_1|)-\mu(|a_2|)=\mu(|b_1|)-\mu(|b_2|)= l \in \mathcal{L},
	$$ 
	therefore $\mu(|[a_1b_1]|)=\mu(|[a_2b_2]|)+l$ for some $l \in \mathcal{L}$.
\end{proof}
It is important to note that the map $\mu$ constructed in Proposition \ref{prop_mu} is by no means unique. In addition to an arbitrary choice of the origin $y_0$, for any vertex $v \in B$ one can shift the value of $\mu$ on its entire preimage $p^{-1}(v)$ by an arbitrary vector $x_v\in \ker(\pi_E)$.
\par
The property (\ref{mu_prop_2}) means that one can factor the map $\mu$ by the action of the translations of the lattice $\mathcal{L}$. Namely, if $q: V \to V/\mathcal{L}$ is the projection of $V$ to its quotient $V/\mathcal{L}=\T^n$ (where $n$ stands for the rank of $\mathcal{L}$), then for any simplices $t_1, t_2 \in \tilde{B}$ such that $p(t_1)=p(t_2)$ their images in $\T^n$ coincide: $q(\mu(|t_1|))=q(\mu(|t_2|))$. In other words, there exists a continuous map $\beta: |B| \to \T^n$, making the following diagram is commutative:
\begin{equation}
\label{square}
\xymatrix{
	|\tilde{B}| \ar[r]_\mu \ar[d]^p& V \ar[d]^q \\
	|B| \ar[r]^\beta & \T^n}
\end{equation}
\begin{proposition}
	For any $d\mbox{-dimensional}$ simplex $t \in \tilde{B}$, $\pi_E(\mu(|t|))$ is a translated copy of the corresponding prototile $\alpha_{p(t)}(p(t))$:
	$$
	\pi_E(\mu(|t|))=\alpha_{p(t)}(p(t))+\tau_t,
	$$
	for some $\tau_t \in E$.
\end{proposition}
\begin{proof}
	The result follows from the application of the property (\ref{mu_prop_1}) to the edges of $t$ and from the definition of the prototile map (\ref{alpha_s}).
\end{proof}
We are now ready to construct the lift of the isometric windings into $V$.
\begin{proposition} \label{phi}
	Let $f: E \to |B|$ be an isometric winding and  $\tilde f$ be a lift of $f$ (see Proposition \ref{lift_f}). Then with an appropriate choice of the origin $y_0$ in the construction of the map $\mu$, the following holds: 
	$$
	\pi_E \circ \mu \circ \tilde f = \mathrm{id}_E.
	$$ 
\end{proposition}
\begin{proof}
    Let $\mathcal{T}$ be the tiling of $E$ defined by the isometric winding $f$ and $a_0 \in E$ one of the vertices of $\mathcal{T}$. Then we can use $\tilde f(a_0)$ as the vertex $v_0$ in the construction of the map $\mu$ in Proposition \ref{def_mu}. Let us choose the origin $y_0 \in V$ in such a way that $\pi_E(y_0) = a_0$. Then $\pi_E(\mu(\tilde f(a_0)))=a_0$, and it remains to show that this identity also holds for any point $x \in E$. Consider a sequence of vertices $a_0,\dots, a_n$ of $\mathcal{T}$ such that any two consecutive elements are connected by an edge of a tile. By the property (\ref{mu_prop_1}) of the map $\mu$
	$$
	\pi_E(\mu(\tilde f(a_i))-\mu(\tilde f(a_{i-1}))) = \rho(p([\tilde f(a_{i-1}), \tilde f(a_i)]))=\rho([f(a_{i-1}), f(a_i)]).
	$$
	By the property (\ref{translated_prototile}) of isometric windings, $\rho([f(a_{i-1}), f(a_i)])=a_i-a_{i-1}$ and since any two vertices of $\mathcal{T}$ can be connected by a sequence of edges, $\pi_E \circ \mu \circ \tilde f$ maps each vertex of $\mathcal{T}$ to itself. On the other hand, by construction of the maps $\tilde f$ and $\mu$, points inside each tile keep their barycentric coordinates, therefore $\pi_E \circ \mu \circ \tilde f = \mathrm{id}_E$.
\end{proof}
The following commutative diagram summarizes the lifting of the simplicial tiling associated with the isometric winding $f$:
$$
\xymatrix{
	E \ar@/_/[ddr]_f \ar[dr]^(.6){\tilde f} \ar[r]^{\mathrm{id}_E}& E &\\
	& |\tilde{B}| \ar[r]_\mu \ar[d]^p& V \ar[d]^q \ar[lu]_{\pi_E}\\
	& |B| \ar[r]^\beta & \T^n}
$$
It should be noted that the choice of the lattice $\mathcal{L}$ (and the projection $\pi_L$) in (\ref{pi_L}) is in general not unique. More precisely, this choice is fixed by the choice of the $\Z\mbox{-submodule}$ $\ker(\lambda)$ in $\ker(\rho_*)$, or equivalently by the choice of the rational subspace
$$
\ker(\lambda \otimes_\Z \Q) \subseteq \ker(\rho_* \otimes_\Z \Q).
$$
The lifting constructed in the case where $\ker(\lambda \otimes_\Z \Q) = 0$ is the maximal (or universal) one, since any other lifting can be obtained as its quotient. The opposite case of the minimal lifting corresponds to the situation when $\pi_L$ is an isomorphism. Let us illustrate these constructions by two examples.
\subsection{Examples}\label{sec:examples}
\subsubsection*{60 degrees rhombi tiling}
The tiling of plane by rhombi with 60 degrees angle is often used to illustrate the idea of lifting, for instance in the context of random tiling models \cite{henley1999random}. Perhaps the most natural way to transform this tiling into a simplicial one consists in cutting each rhombus by its short diagonal. The resulting new edges should be labeled in order to make the operation reversible. This yields 6 different prototiles forming, after gluing together, an FBS-complex $B$ of 6 two-dimensional and 6 one-dimensional cells, and one vertex.
\par
The $\Z\mbox{-module}$ $L$ is generated by the edges of the tiling and thus coincides with the triangular periodic lattice of tile vertices in $E$. Therefore, $\mathrm{rank}(L)=2$ and the minimal lifting corresponds to the degenerated situation where $V=E$. On the other hand, since $|B|$ in this example is homeomorphic to the 2-skeleton of the three-dimensional torus, $H_1(B, \Z)=\Z^3$. Hence, the maximal lifting requires $\dim V=3$ (see e.g., Fig 2c in \cite{henley1999random}).
\subsubsection*{Penrose tiling}
The conventional triangulation of the rhombic Penrose tiling consists in cutting the thin rhombi by the short diagonal and the thick ones by the long diagonal. The result is shown on the Figure \ref{fig:penrose_triangles} together with de Bruijn decorations \cite{de1990updown}. The action of the ten-fold rotational and mirror symmetry generates an orbit of 20 prototiles for each of the two triangles. The resulting FBS-complex thus has 40 two-dimensional and 40 one-dimensional cells, and 4 vertices.
\begin{figure}[ht]
	\centering
	\includegraphics[width=0.8\linewidth]{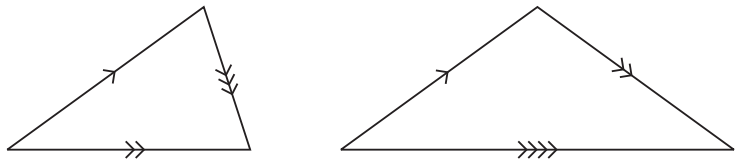}
	\caption{Tiles of the rhombic Penrose tiling cut in halves with de Bruijn decorations.}
	\label{fig:penrose_triangles}
\end{figure}
\par
The $\Z\mbox{-module}$ $L$ is generated by the edges of the triangles and has rank 4. Therefore the standard lifting of the rhombic Penrose tiling, using the lattice $\mathcal{L}=A_4$ (see \cite[Example~7.11]{baake2013aperiodic} and the citations therein), is actually the minimal one. 
\par
It is notable that the FBS-complex $B$ obtained by the triangulation shown on the Figure \ref{fig:penrose_triangles} coincides with the Anderson-Putnam complex for the Penrose tiling used in the computation of the cohomology of the tiling space in \cite[Section~2.4.1]{sadun}, yielding in particular $H^1(B, \R)=\R^5$. Therefore, $\mathrm{rank}(H_1(B, \Z))=5$ and the maximal lifting of the Penrose tiling requires a $5\mbox{-dimensional}$ vector space. It should be emphasized that this lifting is {\em not equivalent} to the frequently used non-minimal embedding of the Penrose tiling in $\R^5$ \cite[Remark~7.8]{baake2013aperiodic}. Indeed, since the extra dimension corresponds to the pattern equivariant integral 1-cocycle defined by the single arrows of the de Bruijn decorations, the fifth coordinate of the lifted tiling is given by the Sutherland's arrow counting function \cite{sutherland1986self}. While the fifth coordinate of the non-minimal embedding takes only four different values, the arrow counting function is unbounded (it grows logarithmically with the distance of the tiling plane). 
\subsection{Weak and minimal matching rules}
One of the most intriguing problems in the theory of aperiodic order is that of understanding how constraints on local arrangement of tiles may result in formation of a long range order (periodic, quasiperiodic, limit-periodic, etc) of the entire tiling. There exists two different approaches to the description of such local constraints: the language of {\em local rules} and that of {\em matching rules}. The former operates with local atlases (the sets of allowed finite patches of the tiling \cite{levitov1988local}), while the latter makes use of matching decorations on neighboring tiles, following the original Wang's formulation of the domino problem \cite{wang1961proving}. In tilings of finite local complexity it is always possible to derive matching rules from the local atlas. More precisely, for a given set of local rules, one can decorate the tiles in such a way that  the set of all tilings satisfying these local rules is in one-to-one correspondence with the set of decorated tilings with matching decorations. The na\"ive converse of this statement is not valid, that is, for the set of all tilings satisfying given matching rules, once the decorations are erased, they cannot be generally recovered from the local environments of undecorated tiles only (as a counterexample one can consider the octagonal Ammann-Beenker tiling \cite{katz1995matching}). In this sense, the relation between matching rules and local rules is similar to that between sofic subshifts and subshifts of finite type in the field of symbolic dynamics (see \cite{fernique2016weak} for a discussion). 
\par
It should be emphasized that the distinction between two types of rules is irrelevant for the problem of emergence of long range order in real physical systems such as quasicrystals. Indeed, the shape of tiles and the color of decorations are purely conventional and can be chosen in many different ways as long as the resulting tilings remain {\em mutually locally derivable} with respect to the actual structure of the material. In this respect, the erasure of decorations is not an innocuous operation as it may alter the MLD class of the tiling. On the other hand, it is always possible to encode the erased decorations in deformations of the tiles. Therefore, both languages --- that of local rules and that of matching rules --- are equally efficient as representations of local constraints on the atomic order in solids. However, the description of tilings in terms of FBS-complexes and isometric windings naturally yields the matching type constraints, as illustrated by Propositions \ref{from_f_to_T} and \ref{from_T_to_f}. For this reason we shall use in this paper the terminology of matching rules. 
\par
The matching rules are often classified by their ``strength'', with the idea that stronger rules impose more stringent constraints on the tiling and thus admit 
smaller sets of tilings. Since no matching rules can discriminate between two tilings which are locally indistinguishable, the strongest possible rules are the ones admitting the tilings belonging to precisely one class of local indistinguishability (see \cite[Chap.~5]{baake2013aperiodic} for definition). Following \cite{socolar1990weak}, such rules, if they exist for a given tiling, are called {\em perfect matching rules}. This terminology suggests that all other rules are somewhat less perfect and therefore flawed. We argue, however, that from the physical point of view using the strength of matching rules as a {\em figure of merit} does not make much sense. Indeed, so far no direct experimental evidence has been given in favor of the hypothesis that atomic interactions in real quasicrystals somehow ``implement'' perfect matching rules. The perfect models are thus preferred because they are thought to be conceptually simpler, in the same sense as a perfect periodic lattice is simpler than the arrangement of atoms in a real crystals. However, this abstract argument can be turned around, since the existing models of perfect matching rules are themselves quite sophisticated. More importantly, this sophistication may be unnecessary to explain the main distinctive feature of quasicrystals --- their diffraction pattern.
\par
Historically, the first models of quasicrystalline structures with realistically looking diffraction patterns were obtained by the cut-and-project schemes \cite{elser1985indexing,duneau1985quasiperiodic,kalugin19850}. If a lifted simplicial tiling fits in such a scheme, the natural candidate for the ``inner'' space is the kernel of the projection $\pi_E$, which we denote by $F$:
$$
\xymatrix{F \ar[r]^{\iota_F} & V \ar[r]^{\pi_E} & E}.
$$
To complete the construction of the cut-and-project scheme \cite[Chap.~7.2]{baake2013aperiodic}, one needs a projection $\pi_F: V \to F$ splitting the above short exact sequence:
\begin{equation}
\label{split}
\xymatrix{
	F\ar@/^/[r]^{\iota_F} & V  \ar@/^/[r]^{\pi_E} \ar@/^/^{\pi_F}[l] & E \ar@/^/^{\iota_E}[l]
	},
\end{equation}
where $\pi_E \circ \iota_E = \mathrm{id}_E$ and $\pi_F \circ \iota_F = \mathrm{id}_F$, in other words
$$
V=E \oplus F
$$
With these notations, Theorem 9.4 of \cite{baake2013aperiodic} immediately yields the following:
\begin{proposition}\label{diffr}
	Let $B$ be a $d\mbox{-dimensional}$ FBS-complex. Then, if for an isometric winding $f: E \to B$ and a point $x \in B$, the set $f^{-1}(x)$ is a regular model set, the diffraction measure of $f^{-1}(x)$ is a pure point measure of the form
    \begin{equation}
    \label{difmeasure}
	\sum_{l^* \in \mathcal{L}^*} I(l^*)\delta_{\iota_E^\top(l^*)}
    \end{equation}
\end{proposition}
The coefficients $I(l^*)$ in (\ref{difmeasure}) give the intensities of the Bragg peaks in the diffraction pattern, which are located at the following positions in the reciprocal space:
\begin{equation}
\label{kmodule}
\mathbf{k}_{l^*} = \iota_E^\top(l^*)\qquad \text{ for } l^* \in \mathcal{L}^*. 
\end{equation}
It is worth remarking here that the positions of the Bragg peaks in (\ref{kmodule}) reflect the global order of the model set defined by the actual lattice $\mathcal{L}$ and the splitting (\ref{split}) used in its construction, while the $\Z\mbox{-module}$ $L=\rho_*(H_1(B))$ depends on the local metric properties of the tiles. The experimental diffraction data provide us primarily with the values of $\mathbf{k}_{l^*}$, thus effectively fixing the lattice $\mathcal{L}$ and making the freedom of choice of the lifting (minimal vs maximal) irrelevant in real physical applications. 
\par
For an isometric winding $f$ satisfying the conditions of Proposition \ref{diffr}, the points of the set $\mu(\tilde f(f^{-1}(x))) \subset V$ are located within a finite distance of the subspace $\iota_E(E) \subset V$. What can be said about the diffraction measure of $f^{-1}(x)$ if it is not a regular model set, but this condition still holds? Contrarily to a common belief, besides the possible occurrence  of continuous or singular continuous contribution, the pure point part of the diffraction measure may also be altered. Indeed, the lifted tiling may exhibit a different periodicity than the lattice $\mathcal{L}$, in which case the diffraction measure will contain Bragg peaks not belonging to the module (\ref{kmodule}). One can imagine an even more complicated situation, for instance, a limit-periodic pattern on the top of the underlying quasiperiodic tiling. However, none of these phenomena is observed experimentally. What is observed, however, is that under some circumstances (e.g., rapid growth), the Bragg peaks (\ref{kmodule}) are slightly shifted in a way compatible with a slight variation of the slope of the subspace $\iota_E(E)$ \cite{nagao2015experimental}. Such variation is known in the physical literature as a ``phason strain''. This phenomenon is related to the observation of the so-called approximant phases. These crystalline phases have chemical composition and local environments very similar to those of the parent quasicrystal. It is often possible to describe their structure by the cut-and-project method with a rational slope of the cut \cite{quivy1996cubic}, in this sense they may be considered as an extreme manifestation of the phason strain. And conversely, experimentalists assess the quality of the quasicrystalline specimens by the diffraction pattern showing no visible phason strain. 
\par
The cut-and-project schemes do not lend themselves well to the description of a non-uniform phason strain. A more appropriate way consists in considering the lifted tiling as a graph of the function $\varphi: E \to F$ defined as:
\begin{equation}
\label{phason}
\varphi = \pi_F \circ \mu \circ \tilde f
\end{equation}
In the physical literature, the function $\varphi$ is often referred to as a local phason coordinate. The phason strain is characterized by the behavior of $\varphi$ on the scale of distances much larger than the size of tiles. In contrast, the short-range features of $\varphi$ are physically irrelevant, since they depend on arbitrary choices made in the construction of the map $\mu$ in Proposition \ref{prop_mu}. On the small scale, we shall only need the following technical property of $\varphi$: 
\begin{proposition}\label{lipschitz}
	The function $\varphi$ is uniformly Lipschitz continuous. 
\end{proposition}
\begin{proof}
	Recall that $\mu$ is defined first on the vertices of $|\tilde{B}|$ and then interpolated on the interior of the simplices by barycentric coordinates. Therefore, $\varphi$ is affine within each tile, and by the property (\ref{mu_prop_2}) of $\mu$, the linear term of $\varphi$ depends only on the type of the corresponding prototile. Since $\varphi$ is continuous and the number of prototiles is finite, $\varphi$ is uniformly Lipschitz continuous.
\end{proof}
\par 
We will be mostly interested in the situation where the geometry of the FBS-complex {\em locks the average slope} of the graph of $\varphi$ for all possible isometric windings: 
\begin{definition}
	A $d\mbox{-dimensional}$ FBS-complex $B$ is said to {\em represent minimal matching rules} if it admits an isometric winding and there exists a linear injective map $\iota_E: E \to V$ right-splitting the sequence (\ref{split}) such that for any isometric winding the phason coordinate (\ref{phason}) grows slower than linearly, that is $\|\varphi(x)\|=o(\|x\|)$ as $\|x\| \to \infty$ (for an arbitrarily chosen norm on the finite-dimensional vector space $F$).
\end{definition}
It is instructive to compare the case of minimal matching rules with the random tiling models. There exist strong arguments  \cite{henley1999random} in favor of the hypothesis that the phason coordinate for a typical random tiling remains bounded for $\dim E > 2$. However, the ensemble of random tilings always contains elements with linearly growing $\varphi$ (for instance, periodic tilings), although their weight tends to 0 in the thermodynamic limit. In contrast, in the model with minimal matching rules, $\varphi$ grows sublinearly for {\em every} allowed tiling. 
\par
The notion of weak matching rules, defined in \cite{levitov1988local} in the context of canonical projection tilings, is naturally generalized to the case of simplicial tilings:
\begin{definition}
	A $d\mbox{-dimensional}$ FBS-complex $B$ is said to {\em represent weak matching rules} if it represents minimal matching rules and for any isometric winding the phason coordinate (\ref{phason}) is globally bounded.
\end{definition}
\section{Main results}\label{sec:main}
The central result of this article consists in establishing a connection between matching rules represented by an FBS-complex $B$ and the properties of the map $\beta: |B| \to \T^n$, more specifically those of the corresponding direct map of homology groups:
$$
\beta_*: H_*(B) \to H_*(\T^n).
$$ 
We shall use the same symbol $\beta_*$ for the maps of homology groups with real coefficients, as long as this will not lead to confusion. We shall also use the natural linear map of $d\mbox{-multivectors}$ of $V^*$ to the space of differential $d\mbox{-forms}$ on $\T^n$:
$$
\gamma^\sharp: \bigwedge\nolimits^d V^* \to \Omega^d(\T^n)
$$
defined in such a way that $q^* \circ \gamma^\sharp$ maps an element of $\bigwedge^d V^*$ to the corresponding constant $d\mbox{-form}$ on $V$. Considering de Rham cohomology class of this form yields the canonical isomorphism:
\begin{equation}
\label{external}
\gamma^*: \bigwedge\nolimits^d V^* \to H^d(\T^n, \R),
\end{equation}
as well as its transpose
\begin{equation}
\label{gamma_star}
\gamma_*: H_d(\T^n, \R) \to \bigwedge\nolimits^d V
\end{equation}
(here we use the natural duality between homology and cohomology groups with coefficients in a field).
\par
The isomorphism $\gamma^*$ allows one to define the element of $H^d(\T^n, \R)$ corresponding to the volume form $\Omega_E \in \bigwedge^d(E^*)$. Since $E$ is an Euclidean space, it is natural to choose $\Omega_E$ normalized with its value on a unit cube equal to 1. Its pullback $\pi_E^\top(\Omega_E) \in \bigwedge^d V^*$ defines then an element $\omega_E \in H^d(\T^n, \R)$: 
$$
\omega_E = \gamma^*(\pi_E^\top(\Omega_E)).
$$ 
\subsection{Necessary and sufficient condition for minimal matching rules}\label{sec:conditions}
\begin{proposition}\label{limit_cycle}
	If an FBS-complex $B$ admits an isometric winding then there exists at least one element $c \in H_d(B, \R)$ satisfying $\omega_E(\beta_*(c))=1$ and represented by a cycle from the closed positive cone $Z_d^+(B, \R)$ in the group of cycles $Z_d(B, \R)$.
\end{proposition}
\par
\begin{proof}
	Let $f: E \to B$ be an isometric winding and $\mathcal{T}$ be the corresponding tiling of $E$. We shall interpret $\mathcal{T}$ as an (infinite) simplicial complex and denote by $f_*$ the corresponding simplicial map to $B$. Consider a sequence of open balls $\mathcal{B}_r \subset E$ for $r \in \N^*$:
	$$
	\mathcal{B}_r = \left\{x \in E \mid \|x\|<r \right\}
	$$
	Let $\mathcal{P}_r \in C_d^+(\mathcal{T})$ denote the sum of all simplices of $\mathcal{T}$ contained entirely within $\mathcal{B}_r$. We shall denote the corresponding union of tiles by $|\mathcal{P}_r| \subset E$ (see Figure \ref{fig:lifting}). Both the volume of the patch $|\mathcal{P}_r|$ and that of its boundary grow with $n$, but the volume grows faster; we shall use this intuitive argument to construct the cycle $c$. Define the $d\mbox{-chain}$ $c_r \in C_d(B, \R)$ as
	\begin{equation}
	\label{c_r}
	c_r= \frac{f_*(\mathcal{P}_r)}{\Omega_E(\mathcal{B}_r)}.
	\end{equation}
	Note that the denominator in (\ref{c_r}) is just the volume of the ball $\mathcal{B}_r$. 
	\begin{figure}[hb]
		\centering
		\includegraphics[width=0.8\linewidth]{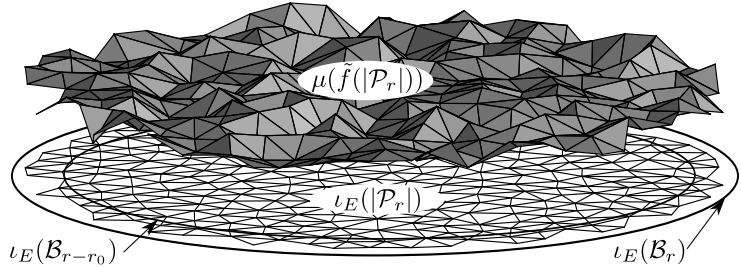}
		\caption{A sketchy representation of the relative position of $\iota_E(E)$ and $\mu(\tilde{f}(E))$ in the case $\dim(V)=3$ and $\dim(E)=2$. The lower part of the figure depicts the image of the tiling patch $|\mathcal{P}_r|$ and the balls $\mathcal{B}_r$ and $\mathcal{B}_{r-r_0}$. The upper part shows the lifted patch $\mu(\tilde{f}(|\mathcal{P}_r|))$.}
		\label{fig:lifting}
	\end{figure}
	\par
	By  Proposition \ref{orientation}, the coefficients of the chain $f_*(\mathcal{P}_r) \in C_d^+(B)$ are non-negative integers, which are equal to the number of copies of the corresponding prototiles in $\mathcal{P}_r$. Since the coefficients of the chains $c_r$ are uniformly bounded (by the biggest inverse volume of prototiles), the sequence $(c_r)$ has an accumulation point $c$ in the standard topology of the finite-dimensional vector space $C_d(B, \R)$. Let us show that $c$ is a cycle. Denote by $\|\cdot\|_1$ the $\ell^1$ norm on the groups of chains $C_*(B, \R)$ and $C_*(\mathcal{T}, \R)$, defined with respect to the basis of simplices. The linear map $f_*$ is contracting with respect to this norm, therefore 
	$$
	\|\partial c_r \|_1 \le \frac{\|\partial \mathcal{P}_r \|_1}{\Omega_E(\mathcal{B}_r)}.
	$$
	On the other hand, $\|\partial \mathcal{P}_r\|_1 \sim r^{d-1}$ and $\Omega_E(\mathcal{B}_r)\sim r^d$, hence $\lim_{r \to \infty} \|\partial c_r \|_1 =0$. Since $C_d(B, \R)$ is a finite dimensional space, the boundary operator is continuous in the norm topology, therefore $\partial c=0$ and $c$ is a cycle. The coefficients of the chains $c_r$ are non-negative, hence $c \in Z_d^+(B, \R)$. Since $C_{d+1}(B)=0$, the cycle $c$ is a unique representative of its class in $H_d(B, \R)$.	
	\par
	Denote the maximum of diameters of all prototiles by $r_0$. Since every point inside $\mathcal{B}_{r-r_0}$ is covered by a tile from $\mathcal{P}_r$, the following inclusions hold (see Figure \ref{fig:lifting}):
	$$
	\mathcal{B}_{r-r_0} \subset |\mathcal{P}_r| \subset \mathcal{B}_r
	$$
	and the volume of the patch $|\mathcal{P}_r|$ grows asymptotically as that of the ball $\mathcal{B}_r$. Therefore,
	$$
	\omega_E(\beta_*(c))) = \lim_{r \to \infty} \frac{\Omega_E(|\mathcal{P}_r|)}{\Omega_E(\mathcal{B}_r)}=1.
	$$
\end{proof}
\par
The coefficients of the cycle $c$ can be understood according to the formula (\ref{c_r}) as 
the average spatial densities of the corresponding tile species in $\mathcal{T}$. Although these quantities should be well defined for a physically relevant model, the construction of $c$ does not guarantee its uniqueness. Indeed, the result may depend on the choice of the isometric winding $f$ and the sequence $c_r$ may have several accumulation points. We shall see, however, that the condition of minimal matching rules fixes $\beta_*(c)$, even if $c$ is not uniquely defined. This result can be conveniently formulated in terms of the Pl\"ucker embedding of Grassmannian manifolds (see \cite{harris2013algebraic}): 
\begin{equation}
\label{plucker}
\psi: G(d, V) \to \PP\left(\bigwedge\nolimits^d V\right).
\end{equation}
\begin{proposition}\label{decomposable}
	If the FBS-complex $B$ represents minimal matching rules, then for any cycle $c\in H_d(B,\R)$, which is an accumulation point of the sequence (\ref{c_r}), the multivector $\gamma_*(\beta_*(c)) \in \bigwedge^d V$ is decomposable and the corresponding point on the Pl\"ucker embedding of $G(d, V)$ designates the subspace $\iota_E(E) \subset V$. In other words,
	$$
	\gamma_*(\beta_*(c))=\iota_E(x_1 \wedge \dots \wedge x_d)
	$$
	for some vectors $x_1, \dots,x_d \in E$.
\end{proposition}
\begin{proof}
	Since $E$ is a $d\mbox{-dimensional}$ space, the condition $\omega_E(\xi)=1$ unambiguously defines a multivector
	\begin{equation}
	\label{xi}
	\xi \in \iota_E\left(\bigwedge\nolimits^d E\right) \subset \bigwedge\nolimits^d V.
	\end{equation}
	Consider an element $\omega \in \bigwedge^d V^*$. By the definition (\ref{c_r}) of the chain $c_r$,
	$$
	\omega\left(\gamma_*\left(\beta_*(c_r)\right)\right)=
	\frac{1}{\Omega_E(\mathcal{B}_r)}
    \int_{\beta(f(|\mathcal{P}_r|))} \gamma^\sharp(\omega)=
    \frac{1}{\Omega_E(\mathcal{B}_r)}
    \int_{\mu(\tilde f(|\mathcal{P}_r|))} \omega.
	$$
	\par
	\begin{figure}[ht]
		\centering
		\includegraphics[width=0.8\linewidth]{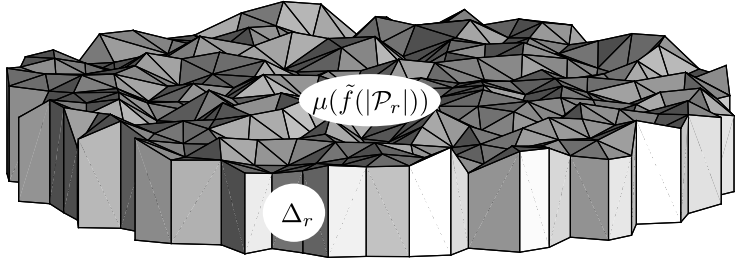}
		\caption{The piecewise-linear unbounded manifold obtained by gluing the collar $\Delta_r$ to the boundaries of $\mu(\tilde{f}(|\mathcal{P}_r|))$ and $\iota_E(|\mathcal{P}_r|)$ (the last part is hidden on the image, but can be seen on Figure \ref{fig:lifting}).}
		\label{fig:pl_manifold}
	\end{figure}
	The idea of the proof is based on the observation that the piecewise-linear spaces $\mu(\tilde f(|\mathcal{P}_r|))$ and $\iota_E(|\mathcal{P}_r|)$ are asymptotically close to each other in the limit $r \to \infty$ and their boundaries can be connected by a relatively small polyhedral ``collar'' $\Delta_r$. More precisely, we define $\Delta_r \subset V$ as the image of the map
	\begin{eqnarray*}
	\partial |\mathcal{P}_r| \times [0, 1] \to V\\
	(x, t) \mapsto \iota_E(x) + t\varphi(x).
	\end{eqnarray*} 
	Since the union
	$$
	\mu(\tilde f(|\mathcal{P}_r|)) \cup \Delta_r \cup \iota_E(|\mathcal{P}_r|)
	$$
	is a piecewise-linear unbounded manifold \cite{rourke2012introduction} of dimension $d$ (see Figure \ref{fig:pl_manifold}),  the following identity holds (with an appropriate choice of orientation on $\Delta_r$):
	$$
	\int_{\mu(\tilde f(|\mathcal{P}_r|))} \omega = 
	\int_{\iota_E(|\mathcal{P}_r|)} \omega +
	\int_{\Delta_r} \omega
	$$
	Since $\sup_{x \in |\mathcal{P}_r|} \|\varphi(x)\|= o(r)$ by the condition of minimal matching rules, the last term in this formula grows with $r$ as $o(r^d)$ and
	$$
	\lim_{r \to \infty} \left(\frac{1}{\Omega_E(\mathcal{B}_r)}
	\int_{\mu(\tilde f(|\mathcal{P}_r|))} \omega \right)=
	\lim_{r \to \infty} \left(\frac{1}{\Omega_E(\mathcal{B}_r)}
	\int_{\iota_E(|\mathcal{P}_r|)} \omega \right)=\omega(\xi).
	$$
	Hence $\omega\left(\gamma_*\left(\beta_*(c)\right)\right)=\omega(\xi)$ for any $\omega \in \bigwedge^d V^*$, therefore
	$$
	\gamma_*\left(\beta_*(c)\right)=\xi.
	$$
\end{proof}
Proposition \ref{decomposable} shows that minimal matching rules impose constraints on the image of the map $\beta_*$. Now, we shall demonstrate that conversely, certain constraints on $\beta_*$ impose the minimal matching rules. 
\begin{definition}
	A subspace $W \subset \bigwedge\nolimits^d V$ is called {\em slope locking} for the splitting (\ref{split}) if it satisfies 
	$$
	T\cdot W=0,
	$$
	where the $d(n-d)\mbox{-dimensional}$ subspace $T\subset \bigwedge\nolimits^d V^*$ is given by the formula
	\begin{equation}
	\label{tilt_forms}
    T=
    \pi_F^\top(F^*) \wedge \pi_E^\top\left(\bigwedge\nolimits^{d-1}E^*\right)
	\end{equation}
\end{definition}
\begin{theorem}\label{logrule}
	If the FBS-complex $B$ admits an isometric winding and the space $\gamma_*(\beta_*(H_d(B, \R)))$ is slope locking for the splitting (\ref{split}), then $B$ represents minimal matching rules with $\|\varphi(x)\|=\mathcal{O}(\log(\|x\|))$. 
\end{theorem}
Before proving this theorem, let us digress and discuss the implications of the slope locking condition. In particular, we have to check that this condition does not exclude physically interesting cases (in particular, those of quasiperiodic tilings). Let us start by justifying the name of the term ``slope locking'':
\begin{proposition}\label{lock}
	If Pl\"ucker coordinates $\xi'$ of a point of $G(d, V)$ annihilate $T$:
	$$
	T\cdot \xi'=0,
	$$
	then the $d\mbox{-dimensional}$ subspace of $V$, corresponding to this point either coincides with $\iota_E(E)$ or is non-transversal with $\iota_F(F)$.
\end{proposition}
\begin{proof}
	Suppose that the subspace corresponding to the point of $G(d, V)$ with Pl\"ucker coordinates $\xi'$ is transversal with $\iota_F(F)$ (such points occupy the cell of maximal dimension in the Schubert decomposition of $G(d, V)$ with respect to the flag containing $\iota_F(F)$, see \cite{rokhlin2003topology}). This subspace can be interpreted as a graph of linear map $E \to F$, and thus can be naturally parameterized by $F \otimes E^*$. Let us choose a basis $(x_1,\dots, x_d)\in E$ normalized by the condition $\Omega_E(x_1 \wedge\dots\wedge x_d)=1$. Then the parameterization in Pl\"ucker coordinates is defined by the map
	\begin{eqnarray}
	\label{xi_prime}
	F\otimes E^* \to \bigwedge\nolimits^d V\nonumber\\
	\eta \mapsto \xi'=\bigwedge_{i=1}^d(\iota_E(x_i)+\iota_F(\eta\cdot x_i)).
	\end{eqnarray}
	Similarly, the space $T$ is parameterized by $F^*\otimes E$ in the following way:
	\begin{eqnarray*}
	F^*\otimes E \to T\\
	\zeta \mapsto (\pi_F^\top \otimes \iota_E)(\zeta)\cdot \omega_E,
	\end{eqnarray*}
    where the dot product stands for the antisymmetrized contraction of $\bigwedge\nolimits^d V^*$ with $V^*\otimes V$. Expanding the expression for $\xi'$ in (\ref{xi_prime}) yields
	$$
	\left((\pi_F^\top \otimes \iota_E)(\zeta)\cdot \omega_E\right)\cdot\xi'= \zeta\cdot \eta.
	$$
	Since $T\cdot \xi'=0$, one has $\zeta\cdot \eta=0$ for all $\zeta \in F^*\otimes E$, which means $\eta=0$. Therefore, the $d\mbox{-dimensional}$ subspace of $V$ corresponding to the point of $G(d, V)$ with Pl\"ucker coordinates $\xi'$ coincides with $\iota_E(E)$.
\end{proof}
An immediate corollary of Proposition \ref{lock} is that a subspace $W \subset \bigwedge\nolimits^d V$ containing a non-zero element of $\iota_E\left(\bigwedge\nolimits^d E\right)$ cannot be slope locking for more than one splitting (\ref{split}). This observation justifies the use of the term --- indeed, the slope locking condition effectively locks the slope of the subspace $\iota_E(E) \subset V$. 
\par
We have also to ensure that the slope locking condition is not too restrictive. The issue is less trivial than merely assessing the codimension of $T$ in $\bigwedge\nolimits^d V^*$. Indeed, since 
$$\beta_*(H_d(B, \R))=\beta_*(H_d(B, \Z))\otimes \R,$$
the subspace $\gamma_*(\beta_*(H_d(B, \R))) \subset \bigwedge\nolimits^d V$ is rational with respect to the lattice $\bigwedge\nolimits^d_\Z \mathcal{L} \subset \bigwedge\nolimits^d V$. Therefore, $\gamma_*(\beta_*(H_d(B, \R)))$ is contained in the annihilator of the rational envelope of $T$ in $\bigwedge\nolimits^d V^*$ (that is, the intersection of all rational subspaces of $\bigwedge\nolimits^d V^*$ containing $T$). This space may be zero even if $T$ is a proper subspace of $\bigwedge\nolimits^d V^*$. Clearly, there are cases when the rational envelope of $T$ has positive codimension --- one such case corresponds to the situation where $\iota_E(E)$ is rational with respect to $\mathcal{L}$. Since this is the case of essentially periodic long range order, the question arises: are there other less trivial solutions?
\par
As follows from Proposition \ref{lock}, for an FBS-complex $B$ satisfying the conditions of Theorem \ref{logrule}, the projective space $\PP(\gamma_*(\beta_*(H_d(B, \R))))$ intersects the Pl\"ucker embedding (\ref{plucker}) of $G(d, V)$ at a discrete set of points, one of which corresponds to the subspace $\iota_E(E) \subset V$. Since $\beta_*(H_d(B, \R)))$ is rational, the Pl\"ucker coordinates of $\iota_E(E)$ must be algebraic over $\Q$. As the Pl\"ucker embedding is surjective for $n<4$, the first case of irrational slope locking occurs for $n=4$ and $d=2$. In this case, the image of $\psi$ in (\ref{plucker}) is a quadric, and the Pl\"ucker coordinates of $\iota_E(E)$ should belong to a quadratic extension of $\Q$. On the other hand, any such subspace can be stabilized by the slope locking condition. Indeed, it suffices to choose $\iota_F(F)$ to be the Galois conjugate of $\iota_E(E)$. Then the 4-dimensional space $T$ coincides with its own conjugate and is therefore a proper rational subspace of the 6-dimensional space $V^*\wedge V^*$. Hence, the annihilator of it rational envelope is non-zero and equals $\iota_E(E \wedge E) \oplus \iota_F(F \wedge F)$. Concrete examples of tilings admitting weak matching rules for this case have been constructed in \cite{bedaride2015periodicities}.
\par
In the general case, the Pl\"ucker coordinates of $\iota_E(E)$ are algebraic numbers of a degree lesser or equal to that of the Pl\"ucker embedding of $G(d, V)$. The latter grows very rapidly with $n$ (see \cite[Chap.~19]{harris2013algebraic}): 
\begin{equation}
\label{deg_plucker}
\deg(G(d,V))=\left(d(n-d)\right)! \prod_{i=0}^{d-1}\frac{i!}{(n-d+i)!}.
\end{equation}
Since the rational envelope of $T$ must include all its Galois conjugates, in the generic situation for $n>4$ and $n-2 \ge d \ge 2$ this envelope coincides with the entire space $\bigwedge\nolimits^d V^*$. However, for some special positions of $\iota_E(E)$ and $\iota_F(F)$ their Pl\"ucker coordinates may have smaller degree than the upper limit (\ref{deg_plucker}). In particular, this may happen when the problem has an additional symmetry. Consider, for instance, the case of the group of rotations of a regular icosahedron acting on $V$ by the sum of the two non-equivalent 3-dimensional real irreducible representations. If the lattice $\mathcal{L} \subset V$ is invariant under this action, the Pl\"ucker coordinates of the two invariant 3-dimensional subspaces of $V$ belong to the quadratic extension $\Q(\sqrt{5})$. In the cut-and-project models of icosahedral quasicrystals one of these subspaces is chosen as the ``physical'' space $\iota_E(E)$ and the other as an ``internal'' space $\iota_F(F)$ \cite{duneau1985quasiperiodic}. Again, both spaces are Galois conjugate of each other, and the annihilator of the rational envelope of $T$ is a 2-dimensional space $\iota_E(\bigwedge\nolimits^3 E) \oplus \iota_F(\bigwedge\nolimits^3 F)$.
\par
Another important case in which the degree of algebraic irrationalities in Pl\"ucker coordinates of $\iota_E(E)$ is lower than the upper limit (\ref{deg_plucker}) is that of planar tilings with $k\mbox{-fold}$ rotational symmetry. The minimal rank of the lattice $\mathcal{L}$ having this symmetry is $n=\phi(k)$ (here $\phi$ stands for the Euler totient function). The space $V$ is then decomposed in the direct sum of invariant 2-dimensional Euclidean planes:
$$
V=\bigoplus_{i=1}^{\phi(k)/2} E_i,
$$
with the corresponding projections defined by their values on the basis $(\ell_j)$ of $\mathcal{L}$:
$$
\pi_{E_i}(\ell_j)=\begin{pmatrix}
\cos\left((2 \pi j u_i)/k\right)\\
\sin\left((2 \pi j u_i)/k\right)
\end{pmatrix},
$$
where $u_i \in (\Z/k\Z)^\times$ is the $i\mbox{-th}$ unit of the ring $\Z/k\Z$. The splitting (\ref{split}) is given by
$$
\iota_E(E)=E_1,\qquad \iota_F(F)=\bigoplus_{i=2}^{\phi(k)/2} E_i.
$$
The rational envelope of $T$ then equals to
$$
\bigoplus_{i<j} E_i^* \wedge E_j^*,
$$
and its annihilator is the rational subspace of $V \wedge V$ of dimension $\phi(k)/2$:
\begin{equation}
\label{annihilator}
\bigoplus_{i=1}^{\phi(k)/2} E_i \wedge E_i.
\end{equation}
Note that in the case $\phi(k)=4$ (that is $k=5$, $8$, $10$ or $12$), the Pl\"ucker coordinates of $\iota_E(E)$ are quadratic irrationalities. Curiously, among quasicrystals experimentally observed so far, all cases of planar symmetry belong to one of these classes (and the only non-planar point symmetry observed in quasicrystals is that of icosahedron). 
\par
To prove Theorem \ref{logrule}, we will need the following technical Lemma:
\begin{lemma}\label{face_average}
	If $B$ is an FBS-complex and the space $\gamma_*(\beta_*(H_d(B, \R)))$ is slope locking for the splitting (\ref{split}), then for any isometric winding of $B$, the difference between the value of the phason coordinate $\varphi$ averaged over a face of a $d\mbox{-dimensional}$ cube in $E$ and that averaged over the opposite face is globally bounded. This bound does not depend on the size of the cube and its position in $E$.
\end{lemma}
\begin{proof}
	Let us start by clarifying the statement of the Lemma. Denote by $\|\cdot\|$ the Euclidean norm in $E$. We shall use the same notation for an arbitrarily chosen norm in $F$ as well as for the norm in the corresponding dual space. Let $\Gamma_r \subset E$ denote an arbitrarily positioned $(d-1)\mbox{-dimensional}$ cube of edge length $r$ in $E$. For $\kappa$ the normal vector to $\Gamma_r$ of unit length, we shall use the notation $[0, r\kappa]$ for the line segment between the points $0$ and $r\kappa$ in $E$. Then the set $\Gamma_r+[0,r\kappa]$ is a $d\mbox{-dimensional}$ cube, having $\Gamma_r$ and $\Gamma_r+r\kappa$ as opposite faces. The value of the phason coordinate $\varphi$ averaged over $\Gamma_r$ is defined as
	$$
	\langle \varphi \rangle_{\Gamma_r}=r^{-d+1}\int_{x \in \Gamma_r} \varphi(x)(\kappa \cdot \Omega_E).
	$$ 
	Note that the interior product $\kappa \cdot \Omega_E$ is the $(d-1)\mbox{-form}$  of volume on the face $\Gamma_r$. The Lemma states that there exists a positive constant $K$ independent on $r$ and on the position of $\Gamma_r$ in $E$, such that for any isometric winding of $B$, the corresponding phason coordinate $\varphi$ satisfies the following inequality:
    \begin{equation}
    \label{dif_faces}
	\|\langle \varphi \rangle_{\Gamma_r+r\kappa} - \langle \varphi \rangle_{\Gamma_r}\|< K.
    \end{equation}
	\par
	To prove the above we shall show that there exists a real number $K>0$ such that for any $\nu \in F^*$ of unit norm the following integral over the $d\mbox{-dimensional}$ cube:
	$$
	I_{\nu\kappa} = \int_{x \in \Gamma_r+[0, r\kappa]} \mathrm{d}(\nu \cdot \varphi(x))\wedge(\kappa\cdot \Omega_E)
	$$
	satisfies the inequality $|I_{\nu\kappa}| < K r^{d-1}$. Indeed, since the partial integration along the direction of $\kappa$ yields
	$$
	\nu\cdot\left(\langle \varphi \rangle_{\Gamma_r+r\kappa} - \langle \varphi \rangle_{\Gamma_r}\right)=r^{-d+1}I_{\nu\kappa},
	$$
	this will prove the statement of the Lemma. The value of $I_{\nu\kappa}$ can also be computed as an integral of a constant $d\mbox{-form}$ in $V$ over the graph of $\varphi$:
	$$
    I_{\nu\kappa}=\int_{\mu(\tilde f(\Gamma_r+[0, r\kappa]))} \omega_{\nu\kappa},
	$$
	where the form $\omega_{\nu\kappa} \in \bigwedge\nolimits^d V^*$ is given by
	\begin{equation}
	\label{omega_tilt}
	\omega_{\nu\kappa} = \pi_F^\top(\nu) \wedge \pi_E^\top(\kappa\cdot \Omega_E)
	\end{equation}
	\begin{figure}[ht]
		\centering
		\includegraphics[width=0.56\linewidth]{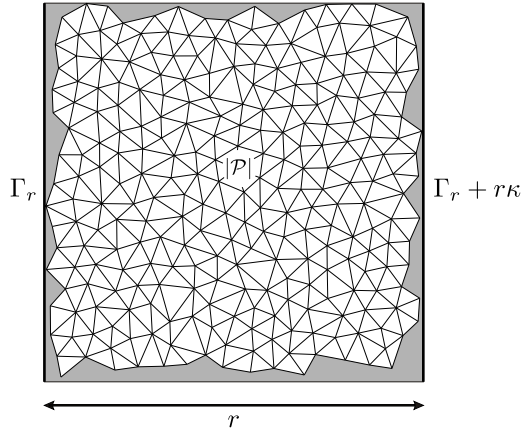}
		\caption{The $d\mbox{-dimensional}$ cube $\Gamma_r+[0, r\kappa]$ containing the patch $|\mathcal{P}|$ of the tiling $\mathcal{T}$. The volume of the shaded space grows with $r$ as $\mathcal{O}(r^{d-1})$.}
		\label{fig:cube}
	\end{figure}
	\par
	Let $\mathcal{T}$ be the tiling corresponding to the isometric winding under consideration. Denote by $\mathcal{P} \in C_d(\mathcal{T})$ the sum of all simplices of $\mathcal{T}$ entirely contained in the cube $\Gamma_r+[0, r\kappa]$ and by $|\mathcal{P}| \subset E$ the union of the corresponding tiles (see Figure \ref{fig:cube}). The volume of the interstice between $|\mathcal{P}|$ and $\Gamma_r+[0, r\kappa]$ grows with $r$ as $\mathcal{O}(r^{d-1})$. Therefore, since $\omega_{\nu\kappa}$ is globally bounded, there exists a constant $K'>0$ such that
	\begin{equation}
	\label{gap_space}
	\left| 
	I_{\nu\kappa} -
	\int_{\mu(\tilde f(|\mathcal{P}|))} \omega_{\nu\kappa}
	\right|< K'r^{d-1}.
	\end{equation}
	The integral in (\ref{gap_space}) can be expressed as a sum over the simplices of $\mathcal{P}$:
	\begin{equation}
	\label{sum_P}
	\int_{\mu(\tilde f(|\mathcal{P}|))} \omega_{\nu\kappa}=
	\sum_{s \in \mathcal{P}}
	\int_{\beta(f(s))}\gamma^\sharp(\omega_{\nu\kappa}).
	\end{equation}
	Note that the form $\omega_{\nu\kappa}$ belongs to the space $T$ defined in (\ref{tilt_forms}). Therefore, by the condition that $\gamma_*(\beta_*(H_d(B, \R)))$ is slope locking, the cochain from $C^d(B, \R)$ defined by the formula
	\begin{equation}
	\label{nameless_cochain}
	B \ni s \mapsto \int_{\beta(s)} \gamma^\sharp(\omega_{\nu\kappa})
	\end{equation}
	is a coboundary and equals to $\mathrm{d}\chi_{\nu\kappa}$ for some other cochain $\chi_{\nu\kappa} \in C^{d-1}(B, \R)$. The cochain (\ref{nameless_cochain}) depends linearly on both $\nu$ and $\kappa$, and one can assume the same for $\chi_{\nu\kappa}$ (since both $C^{d-1}(B, \R)$ and $C^d(B, \R)$  are finite-dimensional vector spaces, the coboundary operator has a right inverse defined on its image). The equation (\ref{sum_P}) then yields
	\begin{equation}
	\label{sum_chi}
	\int_{\mu(\tilde f(|\mathcal{P}|))} \omega_{\nu\kappa}=\sum_{s \in \partial \mathcal{P}} \chi_{\nu\kappa}(f(s)).
	\end{equation}
	Since $B$ contains a finite number of $(d-1)\mbox{-dimensional}$ simplices and both $\nu$ and $\kappa$ have unit norm, all terms of the sum in (\ref{sum_chi}) are globally bounded. The number of these terms grows with $r$ as $\mathcal{O}(r^{d-1})$, therefore there exists $K''>0$ such that
	$$
	\sum_{s \in \partial \mathcal{P}} \chi_{\nu\kappa}(f(s)) < K'' r^{d-1}.
	$$
	This inequality together with (\ref{sum_chi}) and (\ref{gap_space}) yields
	$$
	|I_{\nu\kappa}|< K r^{d-1}
	$$
	for $K=K'+K''$, which proves the Lemma.
\end{proof}
\begin{proof}[Proof of Theorem \ref{logrule}]
	An immediate corollary of Lemma \ref{face_average} is that for any two $d\mbox{-dimensional}$ cubes of edge length $r$ in $E$ sharing a common face, the difference of the phason coordinate $\varphi$ averaged over each cube is bounded by a constant $K>0$ independent on $r$. To show this, we shall use the notations introduced in the proof of Lemma \ref{face_average}. Consider two cubes $\Gamma_r+[-r\kappa, 0]$ and $\Gamma_r+[0, r\kappa]$ sharing the common face $\Gamma_r$. The value of $\varphi$ averaged over the cube $\Gamma_r+[0, r\kappa]$ is defined as
	$$
	\langle \varphi \rangle_{\Gamma_r+[0, r\kappa]} = 
	r^{-d}\int_{x \in \Gamma_r+[0, r\kappa]} \varphi(x) \Omega_E,
	$$
	and can also be computed as
	$$
	\langle \varphi \rangle_{\Gamma_r+[0, r\kappa]} = 
	r^{-1} \int_0^r \langle \varphi \rangle_{\Gamma_r+t\kappa} dt.
	$$
	This expression together with the inequality (\ref{dif_faces}) yields
	\begin{equation}
	\label{adjacent}
	\left\|
	\langle \varphi \rangle_{\Gamma_r+[0, r\kappa]}-
	\langle \varphi \rangle_{\Gamma_r+[-r\kappa, 0]}
	\right\|< K.
	\end{equation}
	\par
	The $d\mbox{-dimensional}$ cube $[0, r]^d$ can be partitioned into $2^d$ cubes of edge length $r/2$. By repeatedly applying the inequality (\ref{adjacent}) to the neighboring cubes of the partition, one obtains the following upper bound for the difference of the value of the phason coordinated averaged over the original cube and that averaged over any of the cubes of the partition (for definiteness, let it be $[0, r/2]^d$):
	\begin{equation}
	\label{nested}
	\left\|
	\langle \varphi \rangle_{[0, r]^d}-
	\langle \varphi \rangle_{[0, r/2]^d}
	\right\|< \frac{Kd}{2}.
	\end{equation}
	Clearly, this upper bound applies as well to partitions of any $d\mbox{-dimensional}$ cube of edge length $r$, independently of its position or orientation. 
	\par
	Let us consider two points $a, b \in E$ such that $\|a-b\|< r$, and find a cube of edge length $r$ containing both of them. For each of these points, one can choose one of $2^d$ cubes of the partition containing it, and subdivide this cube in $2^d$ smaller ones. By applying the procedure recursively, we obtain two sequences of nested cubes, converging towards $a$ and $b$ respectively (see Figure \ref{fig:hierarchy}). At a small enough scale, roughly when the size of the cubes approaches that of the tiles, the uniform Lipschitz continuity (see Proposition \ref{lipschitz}) of $\varphi$ provides a better bound than (\ref{nested}). One can stop the iterative subdivision on this scale and then note that the difference of the phason coordinate at $a$ or $b$ and its value averaged over the smallest cube of the subdivision is bounded by a constant independent on $r$. Since the number of iterations grows as $\log_2(r)$ as $r \to \infty$,
	$$
	\left\|
	\varphi(a)-	\varphi(b)
	\right\| <Kd\log_2(r) + \mathrm{const},
	$$
	which proves Theorem \ref{logrule}.
	\begin{figure}[ht]
		\centering
		\includegraphics[width=0.5\linewidth]{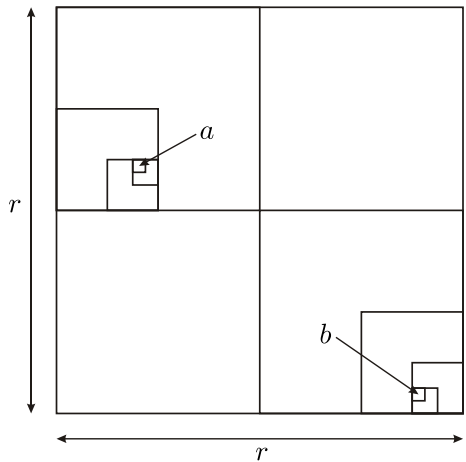}
		\caption{A $d\mbox{-dimensional}$ cube (here $d=2$) of edge length $r$ and two sequences of nested cubes converging to the points $a$ and $b$ (five iterations are shown).}
		\label{fig:hierarchy}
	\end{figure}
\end{proof}
\subsubsection*{Case study: Penrose tilings}
Let us illustrate the results of this section by an example. We shall show that the FBS complex  $B$ of the triangular Penrose tiling with the minimal lifting (see Section \ref{sec:examples}) satisfies the conditions of Theorem 1. Consider the action of the ten-fold symmetry group $\Z/10\Z$ on $B$ and the corresponding representation of this group in $H_2(B, \Q)$. Since $\Z/10\Z$ is finite and $\Q$ is a real field, this representation decomposes into the same classes of irreps as the corresponding contragredient representation in rational cohomologies of $B$. The decomposition of the latter has been computed in \cite[Section~2.4.1]{sadun} (for the Anderson-Putnam complex isomorphic to $B$). As follows from this computation, the rotation by $\pi$ leaves fixed a two-dimensional subspace of $H_2(B, \Q)$, which is also fixed under the action of the entire group $\Z/10\Z$. On the other hand, the rotation by $\pi$ acts by inversion on the $\Z\mbox{-module}$ $L$, and therefore its action on $H_2(\T^4)$ is trivial. Hence, the entire group $\Z/10\Z$ acts trivially on $\gamma_*(\beta_*(H_2(B, \R)))$. Since for the planar rotational symmetry the fixed subspace of $\bigwedge^2 V$ is exactly the annihilator (\ref{annihilator}) of the space $T$ from (\ref{tilt_forms}), $\gamma_*(\beta_*(H_2(B, \R)))$ is slope locking and the condition of Theorem 1 holds.
\par
It is also instructive to see how the slope locking condition breaks if the decorations of the triangles on Figure \ref{fig:penrose_triangles} are erased. This results to an FBS-complex $B$ of 20 triangles, 15 edges and one vertex. The prototiles related by inversion are glued together and the resulting 10 rhombi form a cellular complex isomorphic to the 2-skeleton of the standard cellular decomposition of $\T^5=\R^5/\Z^5$. Therefore, the maximal lifting of an undecorated Penrose tiling corresponds to $\mathcal{L}=\Z^5$, and $\beta_*$ is an isomorphism of $H_2(B)$ and $H_2(\T^5)$. The minimal lifting can be obtained by projecting the maximal one along the direction $(1,1,1,1,1)$ onto the hyperplane orthogonal to this direction (note that this results to $\mathcal{L}=A_4^*$ instead of $A_4$ for the standard Penrose tiling). Thus, in either case $\beta_*$ is surjective and $\gamma_*(\beta_*(H_2(B, \R)))=\bigwedge^2 V$, in violation of the slope locking condition.
\subsection{Defects and robustness}
Real crystals are never perfect and one can safely assume the same for quasicrystals. Therefore, any model of matching rules meant to describe the propagation of the long-range order in real materials, must tolerate defects. In particular, as long as the concentration of defects is small, the long-range order should be affected only slightly. This is a tough challenge for the approach to matching rules based on the Theory of computation \cite{durand2012fixed}. In contrast, as we shall see, the machinery developed in Section \ref{sec:conditions} happens in a natural way to be robust with respect to the presence of defects.
\par
We shall formulate our model of defects in terms of FBS-complexes and isometric windings. More specifically, we define defects in a tiling as special tiles (or groups of tiles, for instance in the case of linear or planar defects), which do not originate from the the ``perfect'' FBS-complex $B$, but belong to some extension $\check{B}\supset B$ of it. On this level of abstraction, one can describe virtually any type of structure imperfections - vacancies, substitution disorder and even dislocations (if the projection of the corresponding Burgers vector from $\mathcal{L}$ on the physical space $E$ is zero). However, since we are mostly interested in the way the defects affect the matching rules, we shall assume that the extended FBS-complex $\check{B}$ allows for the same lifting as does $B$; this will specifically exclude dislocations from the consideration. This requirement can be formulated in the following way. Let $\check{\rho}$ be the homomorphism $C_1(\check{B})\to E$ defining $\check{B}$ as an FBS-complex. Then we shall require that there exist a homomorphism $\check{\lambda}: H_1(\check{B}, \Z) \to \mathcal{L}$ such that the diagram (\ref{pi_L}) can be completed to the following one
$$
\xymatrix{
	H_1(B, \Z) \ar[rr] \ar[rd]^{\rho_*} \ar[rdd]_{\lambda}&& H_1(\check{B}, \Z)\ar[ld]_{\check{\rho}_*}\ar[ldd]^{\check{\lambda}}\\
	& L \\
	& \mathcal{L} \ar[u]_(0.6){\pi_L}
	}
$$
where the homomorphism $H_1(B, \Z) \to H_1(\check{B}, \Z)$ is induced by the inclusion $B \subset \check{B}$. Clearly, an isometric winding of $B$ is also that of $\check{B}$, but the converse is not true. The tiling associated with an isometric winding of $\check{f}: E\to |\check{B}|$ may contain extra tiles, corresponding to the $\check{f}^{-1}(|\check{B}|\backslash |B|)$. We shall refer to these tiles as {\em defective tiles} or simply {\em defects}. 
\par
A physically meaningful description of defects should include a statistical model for their spatial distribution. However, for our purposes, we shall use a rather rudimentary way to describe the defects quantitatively. Let us consider a cube of edge length $r$. The total volume within this cube occupied by defective tiles should grow asymptotically as $\varepsilon r^d$, where $\varepsilon \ll 1$ is the spatial density of defects (see Figure \ref{fig:defects}). We shall require that this volume has an upper bound depending only on $r$ but not on the position or orientation of the cube. At small $r$, this upper bound should be at least $r^d$ as the cube may lie entirely inside a defective tile, and at large $r$ it should grow as $\varepsilon r^d$. The behavior of this bound in the intermediate regime should depend on the nature of the defects (point-like, linear, planar etc). However, for the sake of computational convenience we shall assume the following expression for the upper bound of the volume of defects:
\begin{equation}
\label{bound}
K r^{d-1} + \varepsilon r^d
\end{equation}
One should not seek a deep meaning in this expression. It serves, however, its purpose as long as it dominates any physically meaningful upper bound with appropriate choice of the constant $K$ and has a correct asymptotic behavior as $r \to \infty$.
\par
\begin{figure}[ht]
	\centering
	\includegraphics[width=0.5\linewidth]{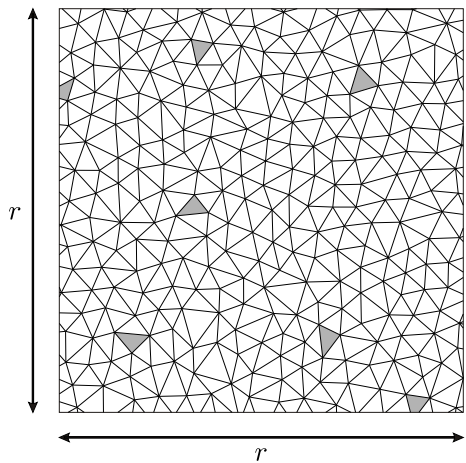}
	\caption{Defects in a simplicial tilings. The shaded area corresponds to the part of $\check{f}^{-1}(|\check{B}|\backslash |B|)$ within a cube of edge length $r$. The volume of this part is bounded by the expression (\ref{bound})}
	\label{fig:defects}
\end{figure}
Let us show now that under these assumption the global phason gradient is bounded by a term proportional to the concentration of defects:
\begin{theorem}\label{robustness}
	Let $B$ be an FBS-complex satisfying the conditions of Theorem \ref{logrule} and $\check{B} \supset B$ an extension of $B$ admitting the same lifting as $B$. Consider an isometric winding of $\check{f}: E\to |\check{B}|$, such that the volume of defective tiles in the corresponding tiling within any cube of edge length $r$ is bounded by the expression (\ref{bound}). Then for any points $a, b \in E$ such that $\|a-b\|<r$, the phason coordinate $\varphi$ corresponding to $\check{f}$ satisfies the inequality
	\begin{equation}
	\label{gradient}
	\left\|
	\varphi(a)-	\varphi(b)
	\right\| <K_1 \log_2(r) + \varepsilon K_2 r + \mathrm{const},
	\end{equation}
	for some constants $K_1$ and $K_2$ independent on $r$ and $\varepsilon$.
\end{theorem}
\begin{proof}
Let us return to the proof of Lemma \ref{face_average}. Since now the patch $|\mathcal{P}|$ may contain defects, the integral over $|\mathcal{P}|$ will contain the contribution of defective tiles. This contribution grows with $r$ at most as the total volume of the defects in $|\mathcal{P}|$, which is bounded by (\ref{bound}). For the part of the patch $|\mathcal{P}|$ free of defects, the reasoning in the proof of Lemma \ref{face_average} remains valid, except that the boundary terms on the right hand side of formula (\ref{sum_chi}) will also contain the boundaries of defective tiles. Since the number of prototiles (\ref{alpha_s}) for $\check{B}$ is finite, the boundary-to-volume ratio for defects is globally bounded and the contribution of the additional boundary terms also grows at most as (\ref{bound}). Therefore, in the presence of defects the formula (\ref{dif_faces}) becomes
$$
\|\langle \varphi \rangle_{\Gamma_r+r\kappa} - \langle \varphi \rangle_{\Gamma_r}\|< K_1' + \varepsilon K_2' r,
$$
for some positive real $K_1'$ and $K_2'$. 
\par
The rest of the proof follows that of Theorem \ref{logrule}, which remains almost unchanged. Inequality (\ref{nested}) now has the form:
$$
\left\|
\langle \varphi \rangle_{[0, r]^d}-
\langle \varphi \rangle_{[0, r/2]^d}
\right\|< \frac{K_1'd}{2} +\frac{\varepsilon K_2' r d}{2}.
$$
Finally, applying this inequality to the sequence of nested cubes on Figure \ref{fig:hierarchy} leads to the upper bound (\ref{gradient}) with $K_1=K_1'd$ and $K_2=2 K_2' d$.
\end{proof}
\subsection{Density of atomic positions} 
A structure model should predict the spatial density of various structure features (for instance, individual atomic species or local environments). As discussed in the Introduction, the natural way to decorate a simplicial tiling consists in placing atoms at the tile's vertices. However, for the sake of generality we shall also consider other kinds of positions, corresponding to an arbitrary point $x \in |B|$. Namely, we shall ask the following question: given an isometric winding $f: E \to B$, what can be said about the density of the set of points $f^{-1}(x)$? In the general setting, speaking of the density of a point set in $E$ implies some sort of averaging with respect to translations. The notion of density is therefore unambiguously defined only if the corresponding dynamical system has a unique ergodic measure. We shall, however, work with a much weaker notion of {\em natural density} \cite[Definition 2.6]{baake2013aperiodic}, defined as the limit value of the average density of points of $f^{-1}(x)$ contained in centered balls of increasing radius:
\begin{equation}
\label{natural}
\lim_{r \to \infty}\frac{\#\left(f^{-1}(x) \cap \mathcal{B}_r\right)}{\Omega_E(\mathcal{B}_r)},
\end{equation}
where $\#$ stands for the cardinality of a set.
\par
Let $\mathcal{T}$ stand for the tiling associated with the isometric winding $f$. If the point $x$ lies within a $d\mbox{-simplex}$ $|s|\subset |B|$, the density of the set $f^{-1}(x)$ equals that of the corresponding tiles in $\mathcal{T}$. For a finite patch $\mathcal{P}_r$ of the tiling, this quantity is given by the corresponding coefficient of the chain $c_r$ in (\ref{c_r}). The density of each tile species is thus well defined if and only if the sequence $c_r$ has a unique accumulation point. The situation becomes more complicated if the point $x$ belongs to a simplex of lesser dimension. In this case, $x$ is effectively shared between several neighboring $d\mbox{-simplices}$. Thus we need to measure the ``fraction'' of the point $x$ excised by a $d\mbox{-simplex}$ $|s|$. This quantity, denoted by $\theta_x(s)$, is formally defined below.
\par
Since $B$ can also be considered as a CW-complex, the inverse of the homeomorphism $\alpha_s$  in (\ref{alpha_s}) can be continued to the closure $\sigma$ of the simplex $\sigma^\circ$. Let us denote this map by
$$
\overline{\alpha_s^{-1}}: \sigma \to |B|.
$$
Consider the set of points
$$
Y_{x,s}=\left\{y \in E \mid \overline{\alpha_s^{-1}}(y)=x\right\}
$$
Clearly, $Y_{x,s}$ is empty if and only if $x \notin \overline{|s|}$. Moreover, according to Proposition \ref{notglued}, $Y_{x,s}$ may contain more than one point only if $x$ is a vertex of $|s|$. Let us define $\theta_x(s)$ as
\begin{equation}
\label{theta_x}
\theta_x(s)= 
\lim_{\delta \to 0^+} 
\sum_{y_i \in Y_{x,s}}
\frac{\Omega_E(\mathcal{B}_{\delta, y_i} \cap \alpha_s(|s|))}{\Omega_E(\mathcal{B}_\delta)},
\end{equation}
where  $\mathcal{B}_{\delta, y_i}$ stands for a ball of radius $\delta$ in $E$ centered at $y_i$ (see Figure \ref{fig:solid_angles}). As a function of $x$, $\theta_x(s)$ can be thought of as a variant of the characteristic function of $|s|$. Indeed, $\theta_x(s)$ is equal to $0$ if $x \notin \overline{|s|}$ and to 1 if $x \in |s|$. If $x$ lies on the boundary of $|s|$, then $\theta_x(s)$ is equal to the sum of the solid angles (as fractions of the full space) occupied by the prototile $\alpha_s(|s|)$ in the vicinity of the points corresponding to $x$ on its boundary.
\begin{figure}[ht]
	\centering
	\includegraphics[width=0.7\linewidth]{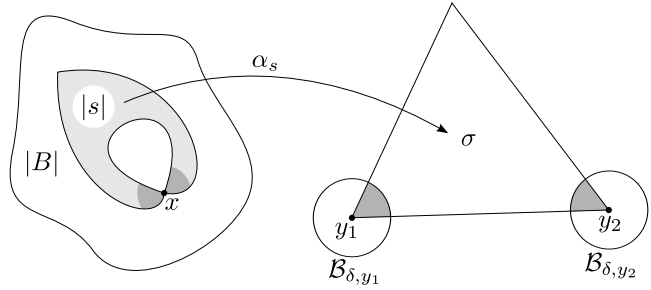}
	\caption{Illustration of the formula (\ref{theta_x}) in the case when $x$ is a vertex of $|s|$. The homeomorphism $\alpha_s$ (see formula (\ref{alpha_s})) takes an open simplex $|s| \subset |B|$ to the interior of an affine simplex $\sigma$. The set $Y_{x,s}$ contains two points $y_1$ and $y_2$. The shaded sectors correspond to the excised balls $\mathcal{B}_{\delta, y_i} \cap \alpha_s(|s|)$ and their preimages in $|B|$.}
	\label{fig:solid_angles}
\end{figure}
\par
One can consider $\theta_x$ as a real-valued function $s \mapsto \theta_x(s)$ defined on the set of $d\mbox{-simplices}$ of $B$, and extend it by linearity to a homomorphism of abelian groups
\begin{equation}
\label{theta_homomorphism}
\theta_x: C_d(B) \to \R.
\end{equation}
\begin{proposition}\label{integer_values}
	$\theta_x$ takes integer values on integral $d\mbox{-cycles}$.
\end{proposition}
\begin{proof}
	Let $z$ be an integral $d\mbox{-cycle}$ on $B$:
	$$
	z=\sum_{s \in B} k_s s, \qquad k_s \in \Z.
	$$
Consider the following open subsets of $E$:
$$
\mathcal{S}_{\delta, y_{i, s}}=\mathcal{B}_{\delta, y_{i,s}} \cap \alpha_s(|s|)-y_{i,s}.
$$
For small enough $\delta$, each $\mathcal{S}_{\delta, y_{i, s}}$ is a sector of the centered $\delta\mbox{-ball}$ $\mathcal{B_\delta}$ cut out by $d$ hyperplanes. Consider the sum of characteristic functions $1_{\mathcal{S}_{\delta, y_{i, s}}}$ of these sectors weighted with $k_s$:
$$
\mathcal{C}_{\delta, x}(y)=\sum_{\substack{{s \in B}\\{y_{i,s} \in Y_{x,s}}}}
k_s\ 1_{\mathcal{S}_{\delta, y_{i, s}}}(y), \qquad y\in E.
$$
Formula (\ref{theta_x}) yields
$$
\theta_x(z)=\lim_{\delta \to 0^+} 
\frac{ \int_{E} \mathcal{C}_{\delta, x}(y)\mathrm{d}y}{\Omega_E(\mathcal{B}_\delta)}.
$$
The function $\mathcal{C}_{\delta, x}$ vanishes beyond the ball $\mathcal{B}_\delta$ and takes integer values inside. The hyperplanes parallel to the faces of the prototiles $\alpha_s(|s|)$ and passing through the origin cut the ball $\mathcal{B}_\delta$ into a finite set of sectors. Within each of these sectors $\mathcal{C}_{\delta, x}$ is constant. On the other hand, should $\mathcal{C}_{\delta, x}$ change its value across a sector boundary, $\partial z$ would contain a non-zero contribution for the simplex of dimension $d-1$ corresponding to this boundary. Since $\partial z=0$, the function $\mathcal{C}_{\delta, x}$ equals the same integer constant on the interior of every sector of $\mathcal{B}_\delta$ and therefore  $\theta_x(z) \in \Z$. 
\end{proof}
\begin{proposition}\label{natural_density}
	If the sequence $c_r$ in (\ref{c_r}) has a unique accumulation point $c$, then for any $x\in |B|$ the natural density of the set $f^{-1}(x)$ equals $\theta_x(c)$.
\end{proposition}
\begin{proof}
	We shall use the notations introduced in the proof of Proposition \ref{limit_cycle}. Summing $\theta_x(s)$ over all simplices $s$ in the $d\mbox{-chain}$ $\mathcal{P}_r$ yields
	$$
	\theta_x(f_*(\mathcal{P}_r))=
	\lim_{\delta \to 0^+}\frac{\Omega_E\left(((f^{-1}(x)\cap \mathcal{B}_r)+\mathcal{B}_\delta)\cap |\mathcal{P}_r|\right)}{\Omega_E(\mathcal{B}_\delta)}
	$$
	The numerator in this formula equals the volume of the union of $\delta\mbox{-balls}$ centered at the points of $f^{-1}(x)$ within $\mathcal{B}_r$ clipped by the patch $|\mathcal{P}_r|$. Therefore, $\theta_x(f_*(\mathcal{P}_r))$ provides a lower bound for $\#\left(f^{-1}(x) \cap \mathcal{B}_r\right)$. Similarly, the upper bound is given by $\theta_x(f_*(\mathcal{P}_{r+r_0}))$ (recall that $r_0$ stands for the maximal diameter of all prototiles):
	$$
	\theta_x(f_*(\mathcal{P}_{r+r_0}))\ge
	\#\left(f^{-1}(x) \cap \mathcal{B}_r\right) \ge
	\theta_x(f_*(\mathcal{P}_r))
	$$
	Then, since the contribution of the chain  $\mathcal{P}_{r+r_0}-\mathcal{P}_r$ is asymptotically negligible in the limit $r \to \infty$, one has for the natural density (\ref{natural})
	\begin{equation}
	\label{limit_density}
	\lim_{r \to \infty} 
	\frac{\#\left(f^{-1}(x) \cap \mathcal{B}_r\right)}{\Omega_E(\mathcal{B}_r)}=
	\lim_{r \to \infty}\theta_x(c_r).
	\end{equation}
	Therefore, since $\theta_x$ is continuous and $\lim_{r\to \infty} c_r=c$, the natural density of $f^{-1}(x)$ equals $\theta_x(c)$
\end{proof}
\par
In periodic structures, the spatial density of atoms is an integer multiple of $|\mathbf{k}_1\wedge\dots \wedge \mathbf{k}_d|$, where $\{\mathbf{k}_i,\,i=1\dots d\}$ is a basis of the reciprocal lattice. The values of $\mathbf{k}_i$ are measured in diffraction experiments, while the atomic density is obtained from the mass density and the chemical composition of the material. Since these measures are completely independent, the relation between the atomic density and the parameters of the reciprocal lattice plays an important role in validation of crystal structures. As has been shown in \cite{kalugin1989density}, under certain conditions of ``matter conservation'', an analogous relation between the density and diffraction wave vectors exists for cut-and-project models of quasicrystals. We shall see that similar formulas can be obtained within the framework of FBS-complexes.
\par
For atomic positions corresponding to a point $x \in |B|$, the {\em matter conservation condition} means that for any two isometric windings $f, f': E \to |B|$ coinciding on all $E$ outside a bounded region, the number of points of the sets ${f}^{-1}(x)$ and ${f'}^{-1}(x)$ contained within this region are equal. Let us consider tilings corresponding to $f$ and $f'$ and denote by $|\mathcal{P}|$ and $|\mathcal{P}'|$ their respective patches within this region. Then the matter conservation condition is equivalent to the requirement that \begin{equation}
\label{conservation}
\theta_x(f_*(\mathcal{P})-f_*'(\mathcal{P}')) = 0.
\end{equation}
The problem of assessing whether a given FBS-complex satisfies the matter conservation condition is difficult and probably undecidable. However, since the cycle $\beta_*(f_*(\mathcal{P})-f_*'(\mathcal{P}'))$ is contractible in $\T^n$, having $\theta_x(\ker(\beta_*))=0$ implies (\ref{conservation}). Let us show that this stronger (but hopefully not excessively strong) condition leads to constraints on the possible values of the atomic density.
\begin{proposition}\label{density_conservation}
	Consider an isometric winding $f: E \to |B|$ of an FBS-complex representing minimal matching rules. If the point $x \in |B|$ is such that $\theta_x(\ker(\beta_*))=0$, then the natural density of the set $f^{-1}(x)$ is well defined and equals $\theta_x(c)$ for any cycle $c\in H_d(B, \R)$ satisfying $\gamma_*(\beta_*(c))=\xi$, where $\xi$ is given by (\ref{xi}) (such cycle exists by virtue of Proposition \ref{decomposable}).
\end{proposition}
\begin{proof}
	Since $\theta_x(\ker(\beta_*))=0$, the value of $\theta_x(c)$ does not depend on the choice of the cycle $c$ satisfying $\gamma_*(\beta_*(c))=\xi$. By Proposition \ref{decomposable}, all accumulation points of the sequence $c_r$  (\ref{c_r}) satisfy this condition. Since $\theta_x$ is continuous on $C_d(B, \R)$, the limit in (\ref{limit_density}) exists and equals $\theta_x(c)$. Therefore, the natural density of $f^{-1}(x)$ is well defined and is equal to $\theta_x(c)$.
\end{proof}
\begin{theorem}\label{theorem_density}
	Consider an FBS-complex $B$ representing minimal matching rules and an isometric winding $f: E \to |B|$. Let $M$ stand for the exponent of the torsion subgroup of $H_d(\T^n, \Z)/\beta_*(H_d(B, \Z))$. If the point $x \in |B|$ is such that $\theta_x(\ker(\beta_*))=0$, then the natural density of the set $f^{-1}(x)$ belongs to the $\Z\mbox{-module}$ generated by
	\begin{equation}
	\label{density_module}
	M^{-1}|\mathbf{k}_{i_1}\wedge \dots \wedge \mathbf{k}_{i_d}|, \qquad 1\le i_1<\dots <i_d\le n,
	\end{equation} 
	where $\{\mathbf{k}_i,\, i=1\dots n\}$ is a basis in a reciprocal quasi-lattice $\iota_E^\top(\mathcal{L}^*)$.
\end{theorem}
\begin{proof} 
	Let us choose a basis $\{\ell_i, \, i=1\dots n \}$ in the lattice $\mathcal{L}$ and decompose the multivector $\xi$ defined by (\ref{xi}) in the corresponding basis of $\bigwedge\nolimits^d V$:
	\begin{equation}
	\label{decomp_xi}
		\xi=
		\sum_{1\le i_1<\dots <i_d\le n}
		m_{i_1\dots i_d}\left(
		\ell_{i_1}\wedge \dots \wedge \ell_{i_d}\right),\qquad m_{i_1\dots i_d} \in \R.
	\end{equation}
	Let $\{\ell_i^* \in \mathcal{L}^*, \, i=1\dots n \}$ stand for the basis reciprocal to $\ell_i$. Then the coefficients $m_{i_1\dots i_d}$ in (\ref{decomp_xi}) are given by
	\begin{equation}
	\label{coeff_m}
	m_{i_1\dots i_d}=\left(
	\ell_{i_1}^*\wedge\dots \wedge \ell_{i_d}^*
	\right)\cdot \xi =
	|\mathbf{k}_{i_1}\wedge \dots \wedge \mathbf{k}_{i_d}|,
	\end{equation}
	where $\mathbf{k}_i=\iota_E^\top(\ell_i^*)$ is the basis of the reciprocal quasilattice $\iota_E^\top(\mathcal{L}^*)$, as in (\ref{kmodule}).
	\par
	Since  $\theta_x$ takes integer values on $H_d(B, \Z)$ by Proposition \ref{integer_values}, one can define the homomorphism 
\begin{equation}
\label{theta_prim}
    \left.\begin{aligned}
	&\theta_x': \beta_*(H_d(B, \Z)) \to \Z \\
	&\theta_x': \beta_*(z) \mapsto \theta_x(z)
	\qquad \forall	z \in H_d(B, \Z)
	\end{aligned}\right.
\end{equation}
	(the last formula is well defined since $\theta_x(\ker(\beta_*))=0$). By passing to the rationals, we can extend $\theta_x'$ to the $\Z\mbox{-module}$ $\Lambda=(\beta_*(H_d(B, \Z)) \otimes \Q) \cap H_d(\T^n, \Z)$. Since $M$ is the exponent of the torsion subgroup of $H_d(\T^n, \Z)/\beta_*(H_d(B, \Z))$, one has $M\Lambda \subset \beta_*(H_d(B, \Z))$. Thus, $\theta_x'(\Lambda) \subset M^{-1}\Z$, and since $\Lambda$ is a direct summand in $H_d(\T^n, \Z)$ we can extend (non-naturally) $\theta_x'$ to
	$$
	\theta_x': H_d(\T^n, \Z) \to M^{-1}\Z. 
	$$
	By Proposition \ref{density_conservation}, the natural density of the set $f^{-1}(x)$ is equal to $\theta_x(c)$ for any $c \in H_d(B, \R)$ satisfying $\beta_*(c)=\gamma_*^{-1}(\xi)$. Then using for $\xi$ the expression (\ref{decomp_xi}) with coefficients (\ref{coeff_m}) and taking into account (\ref{theta_prim}), we get
	$$
	\theta_x(c) = 
	\sum_{1\le i_1<\dots <i_d\le n}
	\left(\theta_x'\left(\gamma_*^{-1}
	\left(\ell_{i_1}\wedge \dots \wedge \ell_{i_d}\right)
	\right)\right)
	|\mathbf{k}_{i_1}\wedge \dots \wedge \mathbf{k}_{i_d}|.
	$$
	Since $\theta_x'\left(\gamma_*^{-1}
	\left(\ell_{i_1}\wedge \dots \wedge \ell_{i_d}\right)
	\right) \in M^{-1}\Z$, this proves the Theorem.
\end{proof}
It is noticeable that the density module obtained in \cite{kalugin1989density} coincides with that given by the formula (\ref{density_module}) for $M=1$. 
\section{Working with real world quasicrystals}\label{sec:real_qc}
\subsection{The strategy}\label{sec:strategy}
In the proposed framework, the structure model of a quasicrystal is thought of as an FBS-complex inferred directly from the experimental data. In this respect our procedure is somehow reciprocal to that of Sections \ref{sec:lifting} and \ref{sec:main}, where the underlying FBS-complex $B$ was assumed to be known, and the properties of the lifted isometric winding had to be found. In contrast, the experimental diffraction data provide us with an information about the lattice $\mathcal{L}$ and the direction of $\iota_E(E)$, and our goal is to construct an FBS-complex representing minimal matching rules. It is worth mentioning that an FBS-complex alone does not make a complete structure model in the traditional sense, since it does not predict the positions of each and every atom in the structure. To be realistic, the model should pass certain validation tests, which are discussed below in Section \ref{sec:validation}. Under favorable circumstances this validation can lead to a more traditional structure model (e.g., a cut-and-project scheme or a model based on inflation), but this is by no means the main goal of the proposed approach.
\par
The diffracted intensity measured in real experiments is a continuous function of the wave-vector, at least because of the finite size of the specimen and the finite instrument resolution. The crucial step in the construction of the structure model thus consists in finding the local maxima of this function and identifying them with points of a finitely generated $\Z\mbox{-module}$ in the reciprocal space (indexing). Once the indexing is done, the real diffraction intensity is replaced by a pure-point measure:
\begin{equation}
\label{difmeasure_exp}
\sum_{l^* \in S \subset \mathcal{L}^*} I(l^*)\delta_{\mathbf{k}_{l^*}},
\end{equation}
where $S$ is a finite subset of points of the reciprocal lattice $\mathcal{L}^*$ and the coefficients $I(l^*)$ are given by the total diffraction intensity integrated around the peak position. This measure is usually thought of as an approximation of the ideal diffraction measure (\ref{difmeasure}). The Fourier transform of this measure is known as the Patterson function. This function arises naturally as a pullback of a periodic function on $V$ by $\iota_E$:
\begin{equation}
\label{patterson}
\sum_{l^* \in S \subset \mathcal{L}^*}
I(l^*) e^{i(\mathbf{k}_{l^*} \cdot \mathbf{x})}=
P(\iota_E(\mathbf{x}))\quad \text{ for } \mathbf{x}\in E,
\end{equation}
where $P: V \to \R$ is defined as
\begin{equation}
\label{patterson_nd}
P(x)=\sum_{l^* \in S \subset \mathcal{L}^*}
I(l^*) e^{i(l^* \cdot x)}.
\end{equation}
The function $P$ is also sometimes referred to as (an $n\mbox{-dimensional}$) Patterson function. An example of such function for the icosahedral phase of $\mathrm{Cd}_{5.7}\mathrm{Yb}$ (obtained with the data of \cite{takakura2007atomic}) is shown on Figure \ref{fig:patterson}.
\begin{figure}[ht]
	\centering
	\includegraphics[width=0.8\linewidth]{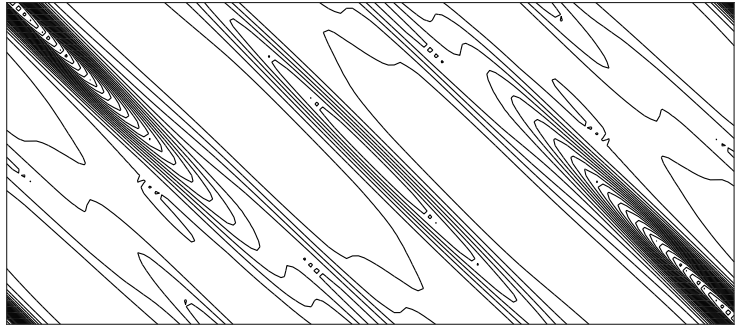}
	\caption{The contour plot of the Patterson function of $\mathrm{Cd}_{5.7}\mathrm{Yb}$ icosahedral quasicrystal. The plot shows a two-dimensional cut of the six-dimensional unit cell along the five-fold symmetry axis.}
	\label{fig:patterson}
\end{figure}
\par
In kinetic diffraction theory, the Patterson function is often equated with the autocorrelation function of the diffracting quantity (the electron density in the case of X-ray diffraction). One should be careful with this assumption since the approximation of the real diffraction intensity by the pure point measure (\ref{difmeasure_exp}) alters its Fourier transform on large distances in a non-controlled way. When working with periodic crystals, this point is often overlooked since the Patterson function is entirely defined by its values within one unit cell. In the case of quasicrystals, however, the function (\ref{patterson}) is aperiodic and one can be tempted to over-interpret the details of its behavior at large distances. Since the main purpose of the considered strategy consists in construction of the FBS-complex from the experimental data, we shall deliberately ignore the behavior of the Patterson function on distances larger than those necessary for this goal.
\par
In periodic crystals, the Fourier transform of the diffracting quantity is itself a pure-point measure carried by the reciprocal lattice. The coefficients of this measure are known as the structure factors and the coefficients of the diffraction measure are proportional the the square of their absolute value. This observation underlies the so-called ``phase problem'' of the classical X-ray crystallography: in fact, in periodic crystals finding the phases of the structure factors amounts to solving the crystal structure. There exist numerous heuristic algorithms for solving the phase problem, all based on the assumption that the spatial distribution of the diffracting quantity should reflect the atomic constitution of the matter. As illustrated by the existence of homometric structures, this consideration alone may not suffice to find an unambiguous answer to the phase problem, in which case other arguments (e.g., realistically looking local environments) should be invoked to select the best solution. Some of the phasing methods do not rely on the explicit assumption of the periodicity of the structure; an example of such method is the charge flipping algorithm \cite{oszlanyi2008charge}. As such, they can be applied to the approximate diffraction measure (\ref{difmeasure_exp}) of quasicrystals. However, since the Fourier transform of the diffracting quantity for quasicrystals is not necessarily a measure (in fact it {\em is not a measure} for the most common structure models), the interpretation of the result of phasing is not as straightforward as in the case of crystals. 
\par
Phasing algorithms take as input the intensities $I(l^*)$ of indexed Bragg peaks $l^*$ belonging to a finite subset $S \subset \mathcal{L}^*$ of the reciprocal lattice and yield the amplitudes $\hat{\varrho}(l^*)$ satisfying
\begin{equation}
\label{varrho_square}
|\hat{\varrho}(l^*)|^2=I(l^*)\quad\text{where } l^* \in S \subset \mathcal{L}^*.
\end{equation}
The inverse Fourier transform of $\hat{\varrho}$ is an $\mathcal{L}\mbox{-periodic}$ function $\varrho: V \to \R$ exhibiting sharp ridges and shallow valleys (see an example on Figure \ref{fig:rho}). The autocorrelation of $\varrho$ equals the $n\mbox{-dimensional}$ Patterson function:
\begin{equation}
\label{conv}
\varrho \circledast \check{\varrho} = P,
\end{equation}
where the symbol $\circledast$ stands for Eberlein convolution \cite[Chap.~8]{baake2013aperiodic} and $\check{\varrho}(x)=\varrho(-x)$. The ridges of $\varrho$ are commonly referred to as ``atomic surfaces'', a term suggestive of a cut-and-project scheme. Indeed, since the autocorrelation function of the pullback $\varrho \circ \iota_E$ equals the $d\mbox{-dimensional}$ Patterson function (\ref{patterson}), it seems natural to interpret $\varrho \circ \iota_E$ as the actual diffracting density and the ridges of $\varrho$ as the characteristic functions of the acceptance domains, inevitably smoothed when approximated by a finite trigonometric sum. At this point the next step in the structure determination often consists in guessing the hidden shape of the atomic surfaces, using as guiding criteria the avoidance of too short interatomic distances, the density and the stoichiometry (see for instance \cite{quiquandon2006unique}). Following this approach while keeping the description of the structure in finite terms, naturally leads to polyhedral atomic surfaces. The resulting models completely predict the position of each and every atom in the structure, but this comes at the cost of rather crude assumptions about the shape of the atomic surfaces. Instead of making such conjectures about the {\em global} structure at this early stage, we shall focus on the study of {\em local} atomic environments, keeping in mind that the mechanism underlying the propagation of long-range quasiperiodic order is somehow encoded by the collection of such environments. In this light, we assign to the phasing algorithm the limited role of inferring the local environments from the autocorrelation function on small distances. We summarize our working hypothesis in the following statement:
\begin{quote}
	Every cluster of atoms of the size of few interatomic distances occurring repetitively in the structure also appears as a cluster of peaks in the function $\varrho \circ \iota_E$. The more frequently the atomic cluster appears in the structure, the higher is the occurrence rate of the corresponding cluster of peaks of $\varrho \circ \iota_E$.
\end{quote}
The striking feature of the plot on Figure \ref{fig:rho} is that the ridges of the function $\varrho$ are flat and parallel to each other. This is a distinctive property of quasicrystals, differentiating them from incommensurate modulated structures. An immediate consequence of this property is that the number of distinct local arrangements of peaks in $\varrho \circ \iota_E$ is (locally) finite. This justifies using models with finite local complexity (in particular tilings) to describe the structure of quasicrystals.
\begin{figure}[ht]
	\centering
	\includegraphics[width=0.8\linewidth]{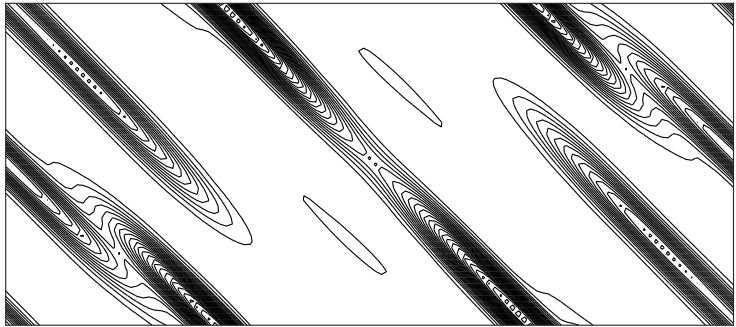}
	\caption{The contour plot of the function $\varrho$ for the icosahedral phase of  $\mathrm{Cd}_{5.7}\mathrm{Yb}$. The plot shows the same two-dimensional cut as Figure \ref{fig:patterson}.}
	\label{fig:rho}
\end{figure}
\par
We suggest a two-stage approach to modeling the structure of quasicrystals by FBS-complexes. At the first stage, the function $\varrho$ is used to construct a large (but finite) ensemble of FBS-complexes, and on the second stage this ensemble is gradually reduced to a single element. It is convenient to think of this ensemble as of a set of all possible subcomplexes of some large FBS-complex $B_0$. The complex $B_0$ should contain enough $d\mbox{-dimensional}$ simplices to allow for the representation of the actual atomic structure by a tiling of $E$ by translated copies of the corresponding prototiles (with atoms located at tile vertices). Each simplex of such tiling corresponds to a cluster of $d+1$ atoms, which, as discussed above, should also appear as a cluster of peaks of $\varrho \circ \iota_E$. This naturally leads to the idea of using the positions of these peaks as vertices of a triangulation of $E$. Note, however, that the spacing between local maxima of $\varrho \circ \iota_E$ may be smaller than the physically acceptable interatomic distance (see e.g., the dumbbell-shaped peak near the label (c) on Figure \ref{fig:cut2}). Such distances should be eliminated before triangulation, by selecting an appropriate subset of peaks of $\varrho \circ \iota_E$. In order to preserve the finite local complexity, the triangulation algorithm should be locally deterministic; in this respect, Delaunay triangulation \cite{fortune1995voronoi} is a reasonable choice. Let $\mathcal{T}_0$ stand for the FLC simplicial tiling of $E$ obtained by such triangulation procedure. At first glance, placing atoms at vertices of $\mathcal{T}_0$ already yields a complete structure model, respecting both the local environments and the diffraction data. This model, however, does not suit our purpose, for the issue of matching rules was completely ignored in its construction. Instead, we shall assume that $\mathcal{T}_0$ results from a surjective isometric winding 
\begin{equation}
\label{f_0}
f_0: E \to |B_0|.
\end{equation}
of the yet unknown FBS-complex $B_0$, and shall use $\mathcal{T}_0$ to build this complex.
\begin{figure}[ht]
	\centering
	\includegraphics[width=0.8\linewidth]{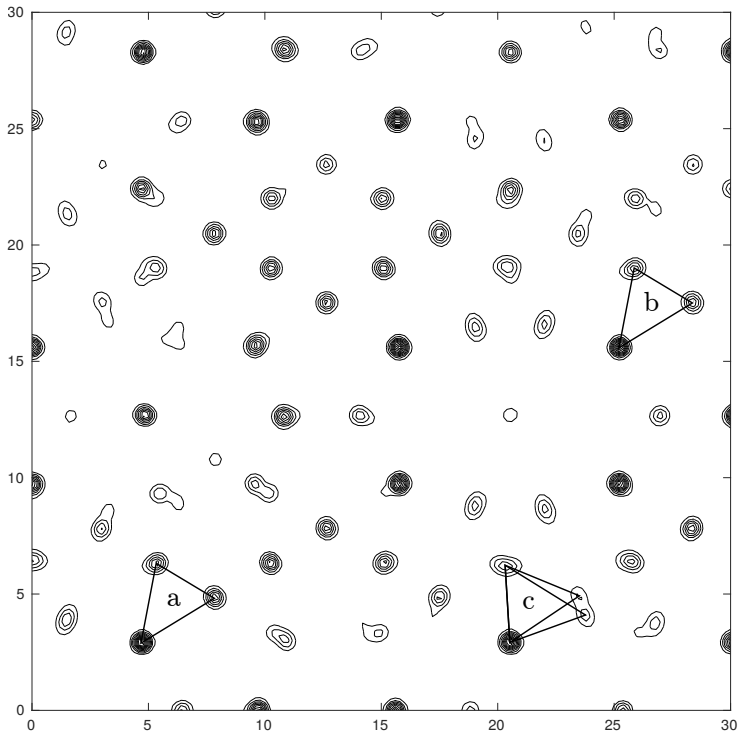}
	\caption{The contour plot of the function $\varrho \circ \iota_E$ for the icosahedral quasicrystal $\mathrm{Cd}_{5.7}\mathrm{Yb}$, restricted to the area $30\text{\AA}\times 30 \text{\AA}$ in the two-fold symmetry plane. Triangles (a) and (b) are $\mathcal{L}\mbox{-equivalent}$. The dumbbell-shaped peak near the label (c) arises from cutting of two different atomic surfaces. The FBS-complex $B_0$ should include 2-simplices corresponding to both of the triangles (c).}
	\label{fig:cut2}
\end{figure}
\par
Every vertex of the tiling $\mathcal{T}_0$ corresponds to a peak of $\varrho \circ \iota_E$, which in turn results from cutting of an atomic surface (a ridge of the function $\varrho$) by $\iota_E(E)$. The lattice $\mathcal{L}$ acts on atomic surfaces by translations. This action has a finite number of orbits, which can be indexed by some finite index set $J$. Let us fix an arbitrary representative atomic surface for each orbit in $J$. Then any other atomic surface of the same orbit can be obtained from this representative by some lattice translation $l \in \mathcal{L}$. In this way, each vertex of $\mathcal{T}_0$ is uniquely indexed by a couple $(j,l) \in J \times \mathcal{L}$ and each $k\mbox{-simplex}$ of $\mathcal{T}_0$ is uniquely defined by the indexes of its $k+1$ vertices. Consider two $k\mbox{-simplices}$ of $\mathcal{T}_0$:
\begin{eqnarray*}
\sigma_1=[(j_{1,0}, l_{1,0}),\dots (j_{1,k}, l_{1,k})]\\
\sigma_2=[(j_{2,0}, l_{2,0}),\dots (j_{2,k}, l_{2,k})].
\end{eqnarray*}
We shall say that $\sigma_1$ and $\sigma_2$ are $\mathcal{L}\mbox{-equivalent}$ and use the notation $\sigma_1 \sim_{\mathcal{L}} \sigma_2$ if there exists $l \in \mathcal{L}$ such that for any $0<i<k$
\begin{eqnarray}
j_{1,i}=j_{2,i}\nonumber\\
l_{1,i}=l_{2,i}+l\label{L_equiv}.
\end{eqnarray} 
It is clear that $\sim_{\mathcal{L}}$ is an equivalence relation. If $\sigma_1 \sim_{\mathcal{L}} \sigma_2$, then the simplices $|\sigma_1|$ and $|\sigma_2|$ are related by a translation from $E$:
$$
|\sigma_1|=|\sigma_2|+\tau,
$$
where $\tau=\pi_E(l)$, with $l$ given by (\ref{L_equiv})  (see an example of triangles labeled (a) and (b) on Figure \ref{fig:cut2}). At this point, we shall make the following hypothesis about the FBS-complex $B_0$:
\begin{quote}
	The isometric winding (\ref{f_0}) maps two tiles of $\mathcal{T}_0$ to the same simplex of $|B_0|$ if and only if these tiles are $\mathcal{L}\mbox{-equivalent}$:
	\begin{equation}
	\label{short_range}
	\sigma_1 \sim_{\mathcal{L}} \sigma_2 
	\iff
	f_0(|\sigma_1|)=f_0(|\sigma_2|).
	\end{equation}
\end{quote}
It should be emphasized that this assumption is far from innocuous. Indeed, in terms of tilings, the condition (\ref{short_range}) limits the possible matching constraints by the ``colored vertex rules'' with colors indexed by the finite set $J$. In physical terms this means that only the positions of the nearest neighboring atoms control the propagation of the long-range quasiperiodic order. It might happen that this restriction is too severe and the proposed algorithm fails to find matching rules of such a short range. We discuss the possible course of action in this case at the end of this section.  
\par
Let us now consider the tiling $\mathcal{T}_0$ as a simplicial complex. Since $\mathcal{L}\mbox{-equivalent}$ simplices have $\mathcal{L}\mbox{-equivalent}$ boundaries, one can define the quotient cellular complex $\mathcal{T}_0/\sim_{\mathcal{L}}$. A class of $\sim_{\mathcal{L}}\mbox{-equivalent}$ simplices can be indexed by a designated representative, for instance, the one with the first vertex belonging to the representative atomic surface:
\begin{equation}
\label{normalize}
[(j_0, l_0), \dots, (j_k, l_k)]
\sim_{\mathcal{L}}
[(j_0, 0), (j_1, l_1-l_0),\dots,(j_k, l_k-l_0)]
\end{equation}
(here the simplex on the right hand side represents the $\mathcal{L}\mbox{-equivalence}$ class of the simplex on the left). The resulting collection is endowed with the structure of a semi-simplicial complex by the following face maps \cite{rourke1971delta}:
\begin{multline}
\label{face_map}
d_i: [(j_0, 0), (j_1, l_1), \dots, (j_k, l_k)]
\mapsto\\
\begin{cases} 
[(j_1, 0), (j_2, l_2-l_1),\dots,(j_k, l_k-l_1)] & \text{if } i=0 \\
[(j_0, 0),\dots, (j_{i-1}, l_{i-1}),(j_{i+1}, l_{i+1}) \dots, (j_k, l_k)]       & \text{if } i\neq 0
\end{cases}
\end{multline}
As follows from the condition (\ref{short_range}), the underlying semi-simplicial complex of $B_0$ in (\ref{f_0}) is isomorphic to $\mathcal{T}_0/\sim_{\mathcal{L}}$. This fact allows us to use the expression $[(j_0, 0), (j_1, l_1), \dots, (j_k, l_k)]$ to index the simplices of $B_0$.  
\par 
We can now complete the construction of $B_0$ as an FBS-complex, and furthermore define the map $\beta: |B_0| \to \T^n$. For each $j \in J$ let us fix a point $y_j \in V$ located on the representative atomic surface of the orbit $j$. The exact position of $y_j$ is constrained by the condition that the subgroup of the space group stabilizing the corresponding atomic surface acts on $y_j$ trivially, while the remaining free parameters of $y_j$ are chosen to maximize the value of $\varrho(y_j)$. Consider the homomorphism $\tilde\rho: C_1(B_0) \to V$ defined by its values on the edges of $B_0$: 
\begin{equation}
\label{tilde_rho}
\tilde\rho:[(j_0, 0), (j_1, l_1)] \mapsto l_1+y_{j_1}-y_{j_0}.
\end{equation}
The composition $\rho=\pi_E \circ \tilde\rho$ satisfies the conditions of Definition \ref{FBS} and provides $B_0$ with the structure of an FBS-complex. Indeed, the property $\rho \circ \partial=0$ is readily verified by checking it on each $2\mbox{-simplex}$ of $B_0$, and the second condition holds since Delaunay triangulation does not produce degenerate simplices (see the ``roundness'' property in \cite{fortune1995voronoi}). To construct the map $\beta: |B_0| \to \T^n$ one can follow the steps of the proof of Proposition \ref{prop_mu} with $\tilde{\rho}$ given by (\ref{tilde_rho}). The resulting map, defined by (\ref{square}) takes the vertex of the type $j$ to the point $q(y_j) \in \T^n$. 
\par
Once the redundant FBS-complex $B_0$ and the map $\beta: |B_0| \to \T^n$ are constructed, one can move on to the search of the candidate structure model among the subcomplexes of $B_0$ admitting minimal matching rules. To qualify for this role, a subcomplex $B \subset B_0$ must satisfy the conditions of Proposition \ref{decomposable} and the slope locking condition of Theorem \ref{logrule} (with the projections $\pi_E$ and $\pi_F$ fixed by the point symmetry of the quasicrystal):
\begin{equation}
\label{bracket}
\iota_E\left(\bigwedge\nolimits^d E \right)
\subseteq \gamma_*(\beta_*(H_d(B, \R))) \subseteq
T^0,
\end{equation}
where $T^0$ stands for the annihilator of the subspace $T\subset \bigwedge\nolimits^d V^*$ defined in (\ref{tilt_forms}). The notations of (\ref{bracket}) make use of the fact that in maximal dimension $H_d(B)$ is naturally a subgroup of $H_d(B_0)$ and that the map of homology groups corresponding to the restriction of $\beta$ to $|B|$ is simply the restriction of $\beta_*$ to $H_d(B)$. It may occur that none of the subcomplexes of $B_0$ satisfies the precondition (\ref{bracket}), in which case the strategy fails right away, indicating that the hypothesis of short-ranged matching rules (\ref{short_range}) may be too restrictive. Otherwise, it remains to find the best candidate among the subcomplexes satisfying (\ref{bracket}), using the validation criteria of Section \ref{sec:validation}. However, the redundant nature of $B_0$ may make testing such subcomplexes one by one prohibitively slow. To improve the efficiency, one can mitigate the all-or-nothing way to select simplices composing $B_0$ by a ``figure-of-merit'' function
\begin{equation}
\label{merit}
\mathcal{M}: \{s \in B_0 \mid \dim(s)=d\} \to \R
\end{equation}
reflecting the likelihood of presence of a given simplex in the structure. This function can be used to implement a heuristic search, in which subcomplexes containing simplices with higher figure-of-merit are tested first.
\par
For reasons of computational efficiency, it is preferable to work with integral homologies. According to (\ref{bracket}), $\gamma_*(\beta_*(H_d(B, \Z)))$ must contain the submodule of $\bigwedge\nolimits^d \mathcal{L}$ belonging to the rational envelope of $\iota_E\left(\bigwedge\nolimits^d E \right)$. Similarly, $\gamma_*(\beta_*(H_d(B, \Z)))$ must be contained within the submodule of $\bigwedge\nolimits^d \mathcal{L}$ belonging to the annihilator of the rational envelope of $T$. It is remarkable that for all known real quasicrystals, both submodules coincide (and have rank 2), thus fixing $\gamma_*(\beta_*(H_d(B, \Z)))$ entirely. The explicit expression for $\gamma_*\circ\beta_*$ is given by the following Proposition:
\begin{proposition}\label{gamma_beta}
	The value of $\gamma_*\circ\beta_*$ on an integral cycle
	\begin{equation}
	\label{coeff_cycle}
	H_d(B_0,\Z) \ni z=\sum_{s\in B_0}
	m_{s} s \qquad\text{ where } m_s \in \Z
	\end{equation}
	is given be the formula
	\begin{equation}
	\label{coeff_gamma_beta}
	\gamma_*(\beta_*(z))=
	\frac{1}{d!}
	\sum_{s\in B_0}
	m_{s}
	\Theta(s)
	\end{equation}
	where
	$$
	\Theta([(j_0,0),(j_1,l_1),\dots,(j_d,l_d)])=
	l_1 \wedge \dots \wedge l_d
	$$
\end{proposition}
\begin{proof}
	For any $\omega \in \bigwedge\nolimits^d V^*$ the value of corresponding de Rham cohomology class $\gamma^*(\omega) \in H^d(\T^n, \R)$ on $\beta_*(z)$ is given by
	$$
	\sum_{s\in B_0}
	m_{s}
	\int_{\beta(s)}\gamma^\sharp(\omega)=
	\frac{1}{d!}
	\sum_{s\in B_0}
	m_{s}
	\left(\omega \cdot \Theta_y(s)\right),
	$$
	where 
	$$
	\Theta_y([(j_0,0),(j_1,l_1),\dots,(j_d,l_d)])=
	(l_1+y_{j_1}-y_{j_0}))\wedge\cdots\wedge
	(l_d+y_{j_d}-y_{j_0})
	$$
	is the exterior product of the vectors (\ref{tilde_rho}) corresponding to the edges of $s$ originating at the vertex $(j_0, 0)$. Therefore, we have
	\begin{equation}
	\label{with_y}
	\gamma_*(\beta_*(z))=
	\frac{1}{d!}
	\sum_{s\in B_0}
	m_{s}
	\Theta_y(s)
	\end{equation}
    Let us now consider the positions $y_j \in V$ as free parameters. Since the face maps (\ref{face_map}) do not depend on $y_j$, the sum (\ref{coeff_cycle}) is always an integral cycle, and $\gamma_*(\beta_*(z)) \in \bigwedge\nolimits^d \mathcal{L}$. Therefore, the right hand side of (\ref{with_y}) is a continuous function of free real parameters $y_j$ taking discrete values and hence constant. Evaluating it for $y_j=0$ we get (\ref{coeff_gamma_beta}).
\end{proof}
\par
The assumption (\ref{short_range}) results in the construction of the FBS-complex $B_0$ with the smallest possible number of vertices, at the cost of losing information about the local configurations in the tiling $\mathcal{T}_0$ on the scale bigger than the size of one tile. If the propagation of the quasiperiodic order relies on local rules of a larger range, this loss may be critical for the considered strategy. In this case one has to reconsider the equivalence relation (\ref{L_equiv}) used in the construction of $B_0$, in order to take into account the local environment of simplices in question. This can be achieved by refining the labels of the vertices of $\mathcal{T}_0$ using, for instance, the following procedure. Let us consider each vertex of $\mathcal{T}_0$ together with its {\em star} (that is the subset of simplices of $\mathcal{T}_0$ incident to this vertex). We shall call two vertex stars $\mathcal{L}\mbox{-equivalent}$ if so are their constituent simplices. There exists a finite number of classes of $\mathcal{L}\mbox{-equivalent}$ vertex stars in $\mathcal{T}_0$, which can be indexed by some finite index set $J'$. Note that this set comes with the natural forgetful map $J'\to J$ taking each $\mathcal{L}\mbox{-equivalence}$ class of vertex stars to the equivalence class of the vertex itself. Let us choose a representative vertex star from each equivalence class. Then each vertex of $\mathcal{T}_0$ is labeled by a unique couple $(j', l) \in J'\times \mathcal{L}$, where $j'$ is the $\mathcal{L}\mbox{-equivalence}$ class of its star, and $l$ is the translation linking this star to the representative of its class as in (\ref{L_equiv}).
\par
The refined vertex labeling can now be used to redefine the $\mathcal{L}\mbox{-equivalence}$ of simplices in (\ref{L_equiv}) and construct a new quotient semi-simplicial complex $B_0'$ with vertices indexed by $J'$. The map of vertices of $|B_0'|$ to those of $|B_0|$ corresponding to the forgetful map $J' \to J$ can be extended to $|B_0'|$ as a surjective semi-simplicial map $\epsilon: |B_0'| \to |B_0|$. The pullback of the homomorphism $\rho:C_1(B_0) \to E$ endows $B_0'$ with the structure of an FBS-complex, making $\epsilon$ into a simplicial FBS-map. The lifting map $\beta': |B_0'| \to \T^n$ is then defined as $\beta'=\beta\circ\epsilon$. Note finally, that this refinement procedure can be applied repeatedly, each time encoding larger local configurations of $\mathcal{T}_0$ by the refined vertex labels.
\subsection{A sketch of the algorithm}\label{sec:sketch}
Let us now implement the strategy outlined in the previous section as a more formal algorithm, with performance considerations in mind. We assume that the phasing is already accomplished and the function $\varrho$ is presented in the form of a finite trigonometric polynomial. The procedure consists in the following steps:
\begin{itemize}
	\item Inventorying the atomic surfaces. This can be achieved e.g., by applying the watershed segmentation algorithm \cite{beucher1992morphological} to the function $\varrho$ directly in its fundamental domain $\T^n$. At this point a human expert decision is required to discriminate the watershed basins containing major ridges of $\varrho$ (those corresponding to the atomic surfaces) and those enclosing minor ripples on the background. Enumerate the atomic surfaces in $\T^n$ by a finite index set $J$ and designate a unique representative for each element of $J$ in the corresponding orbit.
	The result will be presented in the form of an $\mathcal{L}\mbox{-periodic}$ classifying function 
	$$
	V \to (J\times \mathcal{L}) \cup \{*\}
	$$
	associating a point in $V$ with the index of the  atomic surface to which this point belongs, or the singleton $\{*\}$ in the case when the point does not belong to any major ridge of $\varrho$. Associate each atomic surface with a particular atomic species, leaving room for some degree of chemical disorder.
	\item Construction of a large patch $\mathcal{P}_0$ of the tiling $\mathcal{T}_0$:
	\begin{itemize}
		\item Find candidate vertices. Take a large region of the physical space $E$ and find all local maxima of the function $\varrho \circ \iota_E$ in this region. Use the classifying function constructed on the previous step to select only the maxima belonging to the indexed atomic surfaces and label these points with the indices from $J \times \mathcal{L}$. 
		\item Eliminate short distances. In the set of points constructed on the previous step, find all pairs of points separated by a distance shorter than the realistic minimal interatomic spacing (an example of such pair is illustrated by the double peak near the label (c) on Figure \ref{fig:cut2}). Connect these pairs and find connected clusters in the resulting graph. In each cluster, mark randomly selected vertices one at a time in such a way that no two marked vertices are connected by an edge. Once no new vertex can be marked, discard all unmarked vertices .
		\item Apply Delaunay triangulation to the resulting set of points to obtain the patch  $\mathcal{P}_0$. Note that due to the randomness involved in the previous step, similar clusters of peaks in $\varrho \circ \iota_E$ may give rise to different triangulations, as illustrated by two triangles labeled (c) on Figure \ref{fig:cut2}. The resulting diversity of local configurations is in fact a desirable feature of $\mathcal{T}_0$ since it results in a larger FBS-complex $B_0$ and enhances the chances of finding a subcomplex $B \subset B_0$, acceptable as a structure model. 
	\end{itemize}
	\item Construction of the FBS-complex $B_0$ and the map $\beta: |B_0| \to \T^n$. In practice, the procedure described in Section \ref{sec:strategy} is applied to the finite patch $\mathcal{P}_0$ obtained on the previous step, instead of the entire $\mathcal{T}_0$. If $\mathcal{P}_0$ is not large enough, the resulting FBS-complex may contain fewer simplices than $\mathcal{T}_0/\sim_{\mathcal{L}}$. In this case, some of the missing simplices can be recovered by completion of $B_0$ with respect to the action of the point symmetry group. 
	\item Computation of the figure-of-merit function (\ref{merit}). This part of the algorithm is adjustable; one possible choice is to set $\mathcal{M}(s)$ equal to the number of tiles in $\mathcal{P}_0$, belonging to the $\mathcal{L}\mbox{-equivalence}$ class of $s$ (this number should be averaged over the action of the point symmetry group).
	\item Computation of the boundary operator $\partial: C_d(B_0)\to C_{d-1}(B_0)$ in the basis of simplices, using formula (\ref{face_map}) for the face maps.
	\item Preselection of the candidate subcomplexes $B\subset B_0$. We propose two alternative approaches:
	\begin{itemize}
		\item Gradually populate $B$ with simplices in the order of decreasing figure-of-merit (\ref{merit}), adding at each step an entire orbit of simplices with respect to the action of the point symmetry group. After each step, calculate the basis in $H_d(B, \Z)$ and check the condition (\ref{bracket}) using the formula (\ref{coeff_gamma_beta}) for $\gamma_*\circ\beta_*$.
		\item Alternatively, consider the sublattice of integral $d\mbox{-cycles}$ on $B_0$ satisfying the slope-locking condition:
		\begin{equation}
		\label{sublattice}
	    (\gamma_*\circ\beta_*)^{-1}(T^0 \cap \bigwedge\nolimits^d \mathcal{L}).
		\end{equation}
		Note that in all practical cases the annihilator $T^0$ of the subspace $T \subset \bigwedge^d V^*$ (\ref{tilt_forms}) is fixed by the symmetry of the quasicrystal (see Section \ref{sec:conditions}). Vectors of this sublattice having small quadratic norm 
		$$
		\sum_{s\in B_0}m_s s \mapsto \sum_{s\in B_0} \frac{m_s^2}{\mathcal{M}(s)}.
		$$
		correspond to $d\mbox{-cycles}$ on $B$ including small number of simplices with high figure-of-merit (\ref{merit}). A basis of the sublattice (\ref{sublattice}) composed of such vectors can be found by LLL algorithm \cite[Chap.~4]{bremner2011lattice}. To find the candidate subcomplexes $B\subset B_0$, iterate over the combinations of basis vectors with small integer coefficients. For each combination, keep in $B$ only those $d\mbox{-simplices}$ of $B_0$ that enter in the combination with non-zero coefficients, then test the resulting subcomplex for the condition (\ref{bracket}). 
	\end{itemize} 
    In either case, if no candidate subcomplex is found, refine the FBS-complex $B_0$ and start over again.
	\item Validate the candidate models following the steps described in Section \ref{sec:validation}. If none of the candidate subcomplexes admits an isometric winding, refine $B_0$ and start over again.
\end{itemize} 
\subsection{Validation of the model}\label{sec:validation}
Even though the FBS-complex $B_0$ admits at least one isometric winding (\ref{f_0}), the subcomplex $B \subset B_0$ constructed above may not have this property, yet mandatory for the validation of $B$ as a structure model. Since the problem of existence of an isometric winding is undecidable, such validation can only be performed on a case-by-case basis. However, the absence of local obstructions for an isometric winding can be verified in finite time and thus tested systematically for each candidate subcomplex $B \subset B_0$. This can be done by applying to $B$ the reduction procedure of Proposition \ref{reduction}, followed by checking that for the multivector $\xi$ given by (\ref{xi}) the subspace $\beta_*^{-1}(\gamma_*^{-1}(\xi))$ has a non-zero intersection with the positive cone $Z_d^{+}(B, \R)$ (see Propositions \ref{limit_cycle} and \ref{decomposable}).
\par
In some cases an isometric winding of $B$ can be constructed explicitly. This occurs, for instance, if there exists an FBS-map of an inflated copy of $B$ onto $B$ (inflation of an FBS-complex consists in multiplying the homomorphism $\rho$ of Definition \ref{FBS} by a real number larger than 1). One cannot, however, expect the existence of such a map in the general case. On the other hand, when working with a real quasicrystal, we are dealing with a structure that has been literally self-assembled from the melt. One can simulate this process by Monte-Carlo stochastic dynamics, for instance that of the growth model proposed in \cite{joseph1997model}. Clearly, the successful generation of a large patch of tiling respecting the matching rules represented by $B$ does not prove the existence of an isometric winding of $B$, since one can always run into an obstruction requiring the insertion of a defect (e.g., a simplex from $|B_0|\backslash |B|$) at a larger scale. If, however, the density of such defects is small, then according to Theorem \ref{robustness} so is their effect on the long range order, and one can still accept the FBS-complex $B$ as a valid structure model.
\par
If the FBS-complex $B$ satisfies the strong matter conservation condition $\theta_v(\ker(\beta_*))=0$ for any vertex $v$ of $B$ (where the homomorphism $\theta_x$ (\ref{theta_homomorphism}) is defined by (\ref{theta_x})), then by Proposition \ref{density_conservation} the atomic density can be calculated without explicit construction of an isometric winding of $B$. Since the atomic density can be measured separately from the diffraction experiment, this allows for an independent validation of the model. Moreover, Proposition \ref{density_conservation} also predicts the density of atoms corresponding to individual atomic surfaces. This prediction can be verified by comparison with the relative values of these densities which can be inferred independently from the experimental data, assuming that the distribution of atomic species between atomic surfaces is known. Indeed, in the limit of kinetic diffraction theory the corresponding electron densities are proportional to the values of the integral of $\varrho$ over the areas in $\T^n$ corresponding to individual ridges of this function. Such integration can be performed together with the watershed segmentation of the atomic surfaces. 
\par
The conclusive step in the validation of an atomic structure  model is the assay of the agreement of the predicted intensities of Bragg peaks with the experimental data. Unlike cut-and-project schemes, the models based on FBS-complexes do not provide explicit formulas for the diffraction intensities. However, even if no explicit construction of an isometric winding is available, it is still possible to estimate the diffraction intensities numerically from a finite patch of the structure obtained, for instance, by the stochastic growth procedure discussed above.
\section{Conclusions and discussion}\label{sec:conclusion}
We have developed a systematic approach to exploration of matching rules for real quasicrystals, directly in the phased diffraction data. These rules are encoded in a geometric object (an FBS-complex), which can be interpreted as the prototile space of a simplicial tiling. We have also described the specific class of matching rules, for which the propagation of the long-range quasiperiodic order is ensured by the homological properties of the underlying FBS-complex. This class of matching rules is particularly suitable for description of real quasicrystals since the algorithms of computation of the homology groups can handle the inherent uncertainties of the experimental data in a controlled way. Furthermore, in contrast to other known types of matching rules, the homology-based rules are robust with respect to the presence of structure defects. Under some mild assumptions, these rules also predict exact values of the density of individual atomic sites and local environments.
\par
Unlike traditional structure modeling, our approach does not produce immediately the description of the position of each and every atom in a quasicrystal. It may even happen that the discovered matching rules are self-contradictory, that is are impossible to satisfy by any structure filling the entire space. We have proposed some elementary tests for this condition, but it is yet unclear whether these tests will suffice to detect inconsistent matching rules in application to real quasicrystals.
\par
Currently, we are aware of only one example of simplicial tiling possessing homology-based matching rules (the case of the triangular Penrose tiling discussed at the end of Section \ref{sec:conditions}). In the same time, the results of \cite{kalugin2005cohomology} strongly suggest that at least for the decorated Ammann-Beenker and dodecagonal canonical rhombic tilings, the generic Penrose tiling and the icosahedral Ammann-Kramer tiling there exist respectively mutually locally derivable simplicial tilings with homology-based matching rules. It would be interesting therefore to test small FBS-complexes systematically for the condition of Theorem \ref{logrule} in view of possible discovery of entirely new types of tilings with matching rules.
\par
Another open question is whether the logarithmic upper bound of Theorem \ref{logrule} is actually attained. The results of Lemma \ref{face_average} also suggest that the characterization of the long-range quasiperiodic order by the extreme excursions of the phason coordinate may be too rough. 
\par
Finally, it should be pointed out that the construction of the redundant FBS-complex $B_0$ described in Section \ref{sec:strategy} is unrelated to the specific type of matching rules. It can therefore be used for exploration of other possible types of matching rules, under the condition that the subcomplexes of $B_0$ can be tested for these rules in a computationally effective way.
\section*{Acknowledgements}
We are grateful to Denis Gratias and Marianne Quiquandon for stimulating discussions. P.K. also thanks Ekaterina Amerik and Marat Rovinski for clarifying certain aspects of Pl\"ucker embedding.
\bibliographystyle{unsrt}
\bibliography{bof_matching}
\end{document}